\documentclass[letterpaper,11pt]{article}

\pdfoutput=1
 
\usepackage{amsmath}
\usepackage{amsthm}
\usepackage[noend]{algpseudocode}
\usepackage{amsfonts}
\usepackage{graphicx}
\usepackage{wrapfig}
\usepackage{algorithm,algorithmicx,algpseudocode}
\usepackage{lipsum}
\usepackage{url}
\usepackage{comment}
\usepackage{tcolorbox}
\usepackage{hyperref}

\usepackage[margin=1in]{geometry}

\definecolor{mygreen}{RGB}{20,140,80}

\hypersetup{
     colorlinks=true,
     citecolor= mygreen,
     linkcolor= mygreen
}

\algnewcommand\algorithmicforeach{\textbf{for each}}
\algdef{S}[FOR]{ForEach}[1]{\algorithmicforeach\ #1\ \algorithmicdo}
\algblockdefx[ROUND]{Round}{EndRound}[2]{\textbf{round} #1: #2}{\algorithmicend}
\algblockdefx[ROUND]{BeginRound}{EndBeginRound}{\textbf{begin round}}{\textbf{end round}}

\newcommand{\RN}[1]{%
	\textup{\uppercase\expandafter{\romannumeral#1}}%
}
\newcommand{\Td}[0]{\ensuremath{\T{}^d}}

\DeclareMathOperator{\poly}{poly}

\newcommand{\M}[0]{\ensuremath{m}}
\newcommand{\Mset}[0]{\ensuremath{\mathcal{M}}}

\newcommand{\T}[0]{\ensuremath{\mathcal{T}}}
\newcommand{\Tb}[0]{\ensuremath{\T^b}}

\newcommand{\childin}[2]{\ensuremath{\text{child}_{#2}(#1)}}
\newcommand{\child}[1]{\ensuremath{\text{child}(#1)}}

\newcommand{\model}[0]{\ensuremath{\mathsf{MPC}}}

\newcommand{\setsize}{\ensuremath{\alpha}}

\newcommand{\ind}[1]{\ensuremath{#1.\texttt{index}}}
\newcommand{\parent}[1]{\ensuremath{#1.\texttt{parent}}}
\newcommand{\croot}[1]{\ensuremath{#1.\texttt{root}}}

\newcommand{\xhdr}[1]{\vspace{2mm} \noindent{\bf #1}}
\newcommand{\OMIT}[1]{}

\usepackage[suppress]{color-edits}  
\addauthor{sb}{blue}    
\addauthor{md}{red}  
\addauthor{mb}{magenta} 

\newcommand{\eps}[0]{\epsilon}

\renewcommand{\O}[1]{\ensuremath{O(#1)}}
\newcommand{\Ot}[1]{\ensuremath{\widetilde{O}(#1)}}

\DeclareRobustCommand{\mybox}[2][gray!20]{%
\begin{tcolorbox}[
        left=0pt,
        right=0pt,
        top=0pt,
        bottom=0pt,
        colback=#1,
        colframe=#1,
        width=\dimexpr\textwidth\relax, 
        enlarge left by=0mm,
        boxsep=5pt,
        arc=0pt,outer arc=0pt,
        ]
        #2
\end{tcolorbox}
}

\newtheorem{theorem}{Theorem}

\newtheorem{lemma}{Lemma}[section]
\newtheorem{proposition}[lemma]{Proposition}
\newtheorem{corollary}[lemma]{Corollary}

\newtheorem{definition}[lemma]{Definition}
\newtheorem{claim}[lemma]{Claim}
\newtheorem{fact}[lemma]{Fact}

\title{Massively Parallel Dynamic Programming on Trees\footnote{This is the full version of a paper \cite{bateni_et_al:LIPIcs:2018:9166}  appeared at ICALP 2018.}}

\author{
MohammadHossein Bateni\thanks{Google Research, New York. Email: \texttt{\{bateni,mirrokni\}@google.com}.}
\and Soheil Behnezhad\thanks{Department of Computer Science, University of Maryland. Email: \texttt{\{soheil,mahsaa,hajiagha\}@cs.umd.edu}.
Supported in part by NSF CAREER award CCF-1053605,  NSF BIGDATA grant IIS-1546108, NSF AF:Medium grant CCF-1161365,
DARPA GRAPHS/AFOSR grant FA9550-12-1-0423, and another DARPA SIMPLEX grant.
}
\and Mahsa Derakhshan\footnotemark[3]
\and MohammadTaghi Hajiaghayi\footnotemark[3]
\and Vahab Mirrokni\footnotemark[2]
}

\date{}

\begin{document}

\maketitle

\begin{abstract}
Dynamic programming is a powerful technique that is, unfortunately, often inherently sequential. That is, there exists no unified method to parallelize algorithms that use dynamic programming. In this paper, we attempt to address this issue in the {\em Massively Parallel Computations} (\model{}) model which is a popular abstraction of MapReduce-like paradigms. Our main result is an algorithmic framework to adapt a large family of dynamic programs defined over trees. 

We introduce two classes of graph problems that admit dynamic programming solutions on trees. We refer to them as ``($\poly\log$)-expressible'' and ``linear-expressible'' problems. We show that both classes can be parallelized in $O(\log n)$ rounds using a sublinear number of machines and a sublinear memory per machine. To achieve this result, we introduce a series of techniques that can be plugged together. To illustrate the generality of our framework, we implement in $\O{\log n}$ rounds of \model{}, the dynamic programming solution of graph problems such as minimum bisection, $k$-spanning tree, maximum independent set, longest path, etc., when the input graph is a tree. 
\end{abstract}

\clearpage 
\section{Introduction}
With the inevitable growth of the size of datasets to analyze, the
rapid advance of distributed computing infrastructure and platforms
(such as MapReduce, Spark~\cite{spark}, Hadoop~\cite{hadoop},  Flume~\cite{dean2008mapreduce}, etc.), and more
importantly the availability of such infrastructure to medium- and even
small-scale enterprises via services at Amazon Cloud and Google Cloud,
the need for developing better distributed algorithms is felt far and
wide nowadays.  The past decade has seen a lot of progress in studying
important computer science problems in the large-scale setting, which
led to either adapting the sequential algorithms to distributed
settings or at times designing from scratch distributed algorithms for
these problems~\cite{DBLP:conf/soda/AhnGM12, andoni2014parallel,
DBLP:conf/soda/ChitnisCEHMMV16, DBLP:conf/soda/EsfandiariHLMO15,
DBLP:conf/soda/ChitnisCHM15}.


Despite this trend, we still have limited theoretical understanding of
the status of several fundamental problems when it comes to designing
large-scale algorithms.
In fact, even simple and widely used techniques such as the greedy
approach or dynamic programming seem to suffer from an inherent
sequentiality that makes them difficult to adapt in parallel or
distributed settings on the aforementioned platforms. Finding methods
to run generic greedy algorithms or dynamic programming algorithms on
MapReduce, for instance, has broad applications. This is the main goal
of this paper.

\xhdr{Model.} We consider the most restrictive variant of the {\em Massively Parallel
Computations} (\model{}) model which is a common abstraction of MapReduce-like frameworks~\cite{karloff2010model, goodrich2011sorting,
beame2013communication}. Let $n$ denote the input size and let \M{}
denote the number of available machines which is given in the
input.\footnote{As standard, we assume there exists a small constant $0 < \delta < 1$ for which the number of machines is always greater than $n^\delta$.} At each round, every machine can use a space of size $s
= \Ot{n/\M}$ and run an algorithm that is preferably linear time (but
at most polynomial time) in the size of its memory.\footnote{We assume
that the available space on each machine is more than the number of
machines (i.e., $\M \leq s$ and hence $\M = \O{\sqrt{n}}$). It is
argued in \cite{andoni2014parallel} that this is a realistic
assumption since each machine needs to have at least the index of all
the other machines to be able to communicate with them. Also, as argued 
in~\cite{im2017efficient}, in some cases, it is  natural to assume the 
total memory is \Ot{n^{1+\eps}} for a small $0 < \eps < 1$.}  Machines may
only communicate between the rounds, and no machine
can receive or send more data than its memory.

Im, Moseley, and Sun~\cite{im2017efficient} initiated a
principled framework for simulating sequential dynamic programming
solutions of Optimal Binary Search Tree, Longest Increasing
Subsequence, and Weighted Interval Selection problems on this model. This is quite an exciting development, however, it is not
clear whether similar ideas can be extended to other problems and in
particular to natural and well-known graph problems.




In this paper, we give an algorithmic framework that could be
used to simulate many natural dynamic programs on trees. 
Indeed we formulate the properties that make a dynamic program (on trees) amenable to our techniques.  These properties, we show, are natural and are satisfied by many known algorithms for fundamental optimization problems. To illustrate the generality of our framework, we design $\O{\log n}$ round algorithms for well-studied graph problems on trees, such as, minimum bisection, minimum $k$-spanning tree, maximum weighted matching, longest path, minimum vertex cover, maximum independent set, facility location, $k$-center, etc.

\subsection{Related Work}
Though so far several models for
MapReduce have been introduced (see,
e.g., \cite{feldman2010distributing, andoni2014parallel,
goel2012complexity, goodrich2011sorting, kane2010optimal,
pietracaprina2012space, roughgarden2016shuffles}), Karloff,
Suri, and Vassilvitskii \cite{karloff2010model} were the first
to present a refined and simple theoretical model for
MapReduce. In their model, Karloff et al.~\cite{karloff2010model} extract only key features of MapReduce and bypass several message-passing systems and parameters. This model has been extended since then~\cite{beame2013communication, andoni2014parallel, roughgarden2016shuffles} (in this paper we mainly use the further refined model by Andoni, Nikolov, Onak, and Yaroslavtsev~\cite{andoni2014parallel}) and several algorithms both in theory
and practice have been developed in these settings, often using
sketching~\cite{guha2003clustering, kane2010optimal}, 
coresets~\cite{bateni2014distributed, DBLP:conf/stoc/MirrokniZ15} and
sample-and-prune~\cite{kumar2015fast} techniques.  Examples of
such algorithms include $k$-means and
$k$-center clustering~\cite{bahmani2012scalable, bateni2014distributed,
im2015brief}, 
general submodular function
optimization~\cite{ene2015random, kumar2015fast,
mirzasoleiman2013distributed} and query
optimization~\cite{beame2013communication}.
	
Thanks to their common application to data mining tasks,
MapReduce algorithms on graphs have received a lot of attention.
In~\cite{DBLP:conf/spaa/LattanziMSV11} Lattanzi, Moseley, Suri and Vassilvitskii, propose
algorithms for several problems on dense graphs. Subsequently
other authors study graph problems in the MapReduce
model~\cite{DBLP:conf/soda/AhnGM12, ahn2015access,
bahmani2012densest, DBLP:conf/kdd/ChierichettiDK14,
DBLP:conf/kdd/LucierOS15, DBLP:conf/nips/BateniBDHKLM17, DBLP:conf/spaa/BehnezhadDETY17, DBLP:journals/corr/abs-1802-10297} but almost all the solutions apply only
to dense graphs. In contrast in this paper we present a framework
to solve optimization problems on trees, which are important special cases of sparse graphs.

Another related area of research is the design of parallel algorithms. Parallel, and in particular PRAM algorithms, instead of distributing the work-load into different machines, assume a shared memory is accessible to several {\em processors} of the same machine. The \model{} model is more powerful since it combines parallelism with sequentiality. That is, the internal computation of the machines is free and we only optimize the rounds of communication. In this work, we use this advantage in a crucial way to design a unified and general framework to implement a large family of sequential dynamic programs on \model{}. We are not aware of any similar framework to parallelize dynamic programs in the purely parallel settings.

\paragraph{Further developments.} A remarkable series of subsequent works \cite{mistreesmpc,  DBLP:journals/corr/abs-1807-06251,DBLP:journals/corr/abs-1807-08745, concurrent, DBLP:journals/corr/abs-1807-06701} to this paper show how to design even sublogarithmic round algorithms for graph problems such as maximal matching, maximal independent set, approximate vertex cover, etc., using $n^{1-\Omega(1)}$ memory per-machine. This is precisely the regime of parameters considered in this paper. However,  these problems mostly enjoy a {\em locality} property (i.e., they admit fast LOCAL algorithms) that does not hold for the problems considered in this paper.

\subsection{Organization}
We start with an overview of our results and techniques in Section~\ref{sec:results}. We give a formalization of dynamic programming classes in Section~\ref{sec:dpclasses}. Then, we describe a main building block of our algorithms, the tree decomposition method, in Section~\ref{sec:treedecomposition}. Next, in Section~\ref{sec:polylogexpressible}, we show how to solve a particular class of problems which we call $(\poly\log)$-expressible problems. In Section~\ref{sec:linearexpressible} we show how to solve a more general class of problems, namely, linear-expressible problems.

\section{Main Results \& Techniques}\label{sec:results}
We introduce a class of dynamic programming problems which we call $f$-expressible problems. Here, $f$ is a function and we get classes such as $(\poly\log)$-expressible problems or linear-expressible problems. Hiding a number of technical details, the function $f$ equals the number of bits of information that each node of the tree passes to its parent during the dynamic program. Thus, linear-expressible problems are generally harder to adapt to the \model{} model than $(\poly\log)$-expressible problems.

\paragraph{$(\poly\log)$-Expressible problems.} Many natural problems can be shown to be $(\poly\log)$-expressible. For example, the following graph problems are all $(\poly\log)$-expressible if defined on trees: maximum (weighted) matching, vertex cover, maximum independent set, dominating set, longest path, etc. Intuitively, the dynamic programming solution of each of these problems, for any vertex $v$, computes at most a constant number of values. Our first result is to show that every $(\poly\log)$-expressible problem can be efficiently\footnote{Here and throughout the paper, we consider a polylogarithmic round algorithm efficient.} solved in \model{}. As a corollary of that, all the aforementioned problems can be solved efficiently on trees using the optimal total space of \Ot{n}.

\vspace{0.3cm}

\mybox{
\begin{theorem}\label{thm:polylog-expressible}
	For any given $\M$, there exists an algorithm to solve any $(\poly\log)$-expressible problem in \O{\log n} rounds of \model{} using $\M$ machines, such that with high probability each machine uses a space of size at most \Ot{n/\M} and runs an algorithm that is linear in its input size.
\end{theorem}
}

\paragraph{Proof sketch.} The first problem in solving dynamic programs on trees is that there is no guarantee on the depth of the tree. If the given tree has only logarithmic depth one can obtain a logarithmic round algorithm by simulating a bottom-up dynamic program in parallel, where nodes at the same level are handled in the same round simultaneously. This is reminiscent of certain parallel algorithms whose number of rounds depends on the diameter of the graph.

Unfortunately the input tree might be quite unbalanced, with superlogarithmic depth. An extreme case is a path of length $n$.  In this case we can partition the path into equal pieces (despite not knowing a priori the depth of each particular node), handling each piece independently, and then stitching the results together. Complications arise, because the subproblems are not completely independent. Things become more nuanced when the input tree is not simply a path.

To resolve this issue, we adapt a celebrated {\em tree contraction} method to our model. The algorithm decomposes the tree into pieces of size at most $\Ot{m}$ (i.e., we can fit each component completely on one machine), with small interdependence. Omitting minor technical details, the latter property allows us to almost independently solve the subproblems on different machines. This results in a partial solution that is significantly smaller than the whole subtree; therefore, we can send all these partial solutions to a master machine in the next round and merge them.

\paragraph{Linear-Expressible problems.} Although many natural problems are indeed $(\poly\log)$-expressible, there are instances that are not. Consider for example the {\em minimum bisection problem}. In this problem, given an edge-weighted tree, the goal is to assign two colors (blue and red) to the vertices in such a way that minimizes the total weight of the edges between blue vertices and red vertices while half of the vertices are colored red and the other half are blue. In the natural dynamic programming solution of this problem, for a subtree $T$ of size $n_T$, we store \O{n_T} different values. That is, for any $i \in [n_T]$, the dynamic program stores the weight of the optimal coloring that assigns blue to $i$ vertices and red to the rest of $n_T - i$ vertices. This problem is not necessarily $(\poly\log)$-expressible unless we find another problem specific dynamic programming solution for it. However, it can be shown that minimum bisection, as well as many other natural problems, including $k$-spanning-tree, $k$-center, $k$-median, etc., are linear-expressible.

It is notoriously more difficult to solve linear-expressible problems using a sublinear number of machines and a sublinear memory per machine. However, we show that it is still possible to obtain the same result using a more involved algorithm and a slightly more total memory.

\vspace{0.3cm}

\mybox{
\begin{theorem}[Main Result]\label{thm:linear-expressible}
	For any given $\M$, there exists an algorithm to solve any linear-expressible problem that is splittable in \O{\log n} rounds of \model{} using $\M$ machines, such that with high probability each machine uses a space of size \Ot{n^{4/3}/\M}.
\end{theorem}
}

\paragraph{Proof sketch.} Recall that linear-expressibility implies that the dynamic programming data on each node, can be as large as the size of its subtree (i.e., even up to $\O{n}$). Therefore, even by using the tree decomposition technique, the partial solution that is computed for each component of the tree can be linear in its size. This means that the idea of sending all these partial data to one master machine, which worked for $(\poly\log)$-expressible problems, does not work here since when aggregated, they take as much space as the original input. Therefore we have to distribute the merging step among the machines.

Assume for now that each component that is obtained by the tree decomposition algorithm is contracted into a node and call this {\em contracted tree}. The tree decomposition algorithm has no guarantee on the depth of the contracted tree and it can be super-logarithmic; therefore a simple bottom-up merging procedure does not work. However, it is guaranteed that the contracted tree itself (i.e., when the components are contracted) can be stored in one machine. Using this, we send the contracted tree to a machine and design a {\em merging schedule} that informs each component about the round at which it has to be merged with each of its neighbours. The merging schedule ensures that after $\O{\log n}$ phases, all the components are merged together. The merging schedule also guarantees that the number of neighbours of the components, after any number of merging phases, remains constant. This is essential to allow (almost) independent merging for many linear-expressible problems such as minimum bisection.

Observe that after a few rounds, the merged components grow to have up to $\Omega(n)$ nodes and even the partial data of one component cannot be stored in one machine. Therefore, even merging the partial data of two components has to be distributed among the machines. For this to be possible, we use a {\em splitting} technique of independent interest that requires a further {\em splittability} property on linear-expressible problems. Indeed we show that the aforementioned linear-expressible problems, such as minimum bisection and $k$-spanning tree, have the splittability property, and therefore Theorem~\ref{thm:linear-expressible} implies they can also be solved in $\O{\log n}$ rounds of \model{}.


\newcommand{\inp}[1]{#1.\texttt{data}}
\newcommand{\sol}[1]{\ensuremath{\mathcal{D}_{#1}}}

\section{Dynamic Programming Classes}\label{sec:dpclasses}
To understand the complexity of different dynamic programs in the \model{} model, we first attempt to classify these problems. The goal of this section is to introduce a class of problems which we call $f$-expressible problems.

We say a problem $P$ \emph{is defined on trees}, if the input to $P$ is a rooted tree \T{}, where each vertex $v$ of \T{} may contain data of size up to $\poly(\log{n})$, denoted by \inp{v}.\footnote{Consider weighted trees as an example of the problems that define an additional data on the nodes (in case the edges are weighted, each vertex stores the weight of its edge to its parent).}

Dynamic programming is effective when the solution to the problem can be computed recursively. On trees, the recursion is usually defined on the nodes of the tree, where each node computes its data using the data provided by its children. Let \sol{P} denote a dynamic program that solves a problem $P$  that is defined on trees. Also let $P_\T$ be an instance of $P$ and denote its input by \T{}. We use \sol{P}(\T) to denote the solution of $P_\T$ and use \sol{P}(v) to denote the dynamic data that \sol{P} computes for a vertex $v$ of \T{}. The first property that we define is ``binary adaptability''.

\paragraph{Binary Adaptability.} To illustrate when this property holds, consider any vertex $v$ of an input tree and a given ordering of its children $(u_1, u_2, \ldots, u_k)$ and let $\sol{}(u)$ denote the dynamic data of vertex $u$. As long as the operation to compute the dynamic data of $v$ could be viewed as an aggregation of binary operations over the dynamic data that is provided by its children in the given order, i.e., $\sol{}(v) = f_1(\sol{}(u_1), f_2(\sol{}(u_2), \ldots))$, this property is satisfied. (For example, $\max{\{\sol{}(u_1), \sol{}(u_2), \sol{}(u_3)\}}$ could be viewed as $\max\{\sol{}(u_1), \max\{\sol{}(u_2), \sol{}(u_3)\}\}$). Most of the natural dynamic programs satisfy this property.

To formally state this property, we first define ``binary extensions'' of trees and then define a problem to be binary adaptable if there exists a dynamic programming solution that given any binary extension of an input, calculates the solution of the original tree.

\begin{definition}[Binary Extension]
	A binary tree $\Tb$ is a binary extension of a tree $\T$ if there is a one-to-one (and not necessarily surjective) mapping  $f: V(\T) \to V(\Tb)$ such that $f(v)$ is an ancestor of $f(u)$ in $\Tb$ if $v$ is an ancestor of $u$ in $\T$. We assume the data of each vertex $v$ in $T$ is also stored in its equivalent vertex, $f(v)$ in $\Tb$.
\end{definition}

Intuitively, a binary extension of a tree \T{} is a binary tree that is obtained by adding auxiliary vertices to \T{} in such a way that ensures any ancestor of each vertex remains its ancestor.

\begin{definition}[Binary Adaptability]
Let $P$ be a problem that is defined on trees. Problem $P$ is \textit{binary adaptable} if there exists a dynamic program \sol{P} where for any instance $P_\T$ of $P$ with \T{} as its input tree and for any binary extension \Tb{} of \T{}, the output of \sol{P}(\Tb) is a solution of $P_\T$. We say \sol{P} is a binary adapted dynamic program for $P$.
\end{definition}

We are now ready to define $f$-expressiveness. We will mainly consider ($\poly\log$)-expressible and linear-expressible problems in this paper, however the definition is general.

\newcommand{\compressor}[0]{\ensuremath{\mathcal{C}}}
\newcommand{\merger}[0]{\ensuremath{\mathcal{M}}}

\paragraph{$f$-Expressiveness.} This property is defined on the problems that are binary adaptable. Roughly speaking, it specifies how large the data that we store for each subtree should be to be able to merge subtrees efficiently.
\begin{definition}[$f$-Expressiveness] \label{def:Expressiveness}
	Let $P$ be a binary adaptable problem and let \sol{P} be its binary adapted dynamic program. Moreover, let \Tb{} be a binary extension of a given input to $P$. For a function $f$, we say $P$ is $f$-expressible if the following conditions hold.
	\begin{enumerate}
		\item The dynamic data of any given vertex $v$ of \Tb{} has size at most $\Ot{f(n_v)}$ where $n_v$ denotes the number of descendants of $v$.
		\item There exist two algorithms \compressor{} (compressor), and \merger{} (merger) with the following properties. Consider any arbitrary connected subtree $T$ of \Tb{} and let $v$ denote the root vertex of $T$. If for at most a constant number of the leaves of $T$ the dynamic data is unknown (call them the unknown leaves of $T$), algorithm \compressor{} returns a partial data of size at most $\Ot{f(n_T)}$ for $T$ (without knowing the dynamic data of the unknown leaves) in time $\poly(f(n_T))$ such that if this partial data and the dynamic data of the unknown leaves are given to \merger{}, it returns \sol{P}(v) in time $\poly(f(n_T))$.
		\item For two disjoint subtrees $T_1$ and $T_2$ that are connected to each other with one edge, and for subtree $T = T_1 + T_2$, if there are at most a constant number of the leaves of $T$ that are unknown, then $\merger{}(\compressor{}(T_1), \compressor{}(T_2))$ returns $\compressor{}(T)$ in time $\poly(f(n_T))$.
		\end{enumerate}
\end{definition}

We remark that graph problems such as maximum (weighted) matching, vertex cover, maximum independent set, dominating set, longest path, etc., when defined on trees, are all $(\poly\log)$-expressible. Also, problems such as minimum bisection, $k$-spanning tree, $k$-center, and $k$-median are linear-expressible on trees.


\section{Tree Decomposition}\label{sec:treedecomposition}
	The goal of this section is to design an algorithm to decompose a given tree \T{} with $n$ vertices into \O{m} ``components'' of size at most \Ot{n/m}, such that each component has at most a constant number of ``outer edges'' to the other components. To do this, we combine a tree contraction (see \cite{reid1993list, abrahamson1989simple, gazit1988optimal}) method of traditional parallel algorithms with a number of new ideas.

\paragraph{Overview.} To motivate the need to design an involved decomposition algorithm, we first give hints on why trivial approaches do not work. Perhaps the simplest algorithm that comes to mind for this task, is random sampling. That is, to choose a set of randomly selected edges (or vertices) and temporarily remove them so that the tree is decomposed into a number of smaller components. Note that the main difficulty is that we need to bound the size of the largest component with a high probability. It could be shown that this trivial algorithm, even for the simple case of full binary trees fails. More precisely, it could be shown that the size of the component that contains the root is of size $\Omega(n)$ (instead of $\Ot{n/m}$) with a constant probability. In general, naive random partitioning does not perform well.

In contrast, we employ the following algorithm. First, we convert the tree into a binary tree. To do this, we first have to compute the degree of each vertex and then add \emph{auxiliary} vertices to the graph. After that, our algorithm proceeds in \emph{iterations} (not rounds). In each iteration, we merge the vertices of the tree based on some local rules and after \O{\log{n}} iterations, we achieve the desired decomposition. More precisely, in each iteration, we select a set of {\em joint} vertices according to a number of local rules (e.g., the degree of a vertex, status of its parent, etc.) and \emph{merge} each vertex to its closest selected ancestor. We prove in fact these local rules guarantee that the size of the maximum component is bounded by \Ot{n/m} with high probability.

The input, as mentioned before, is a rooted tree \T{}, with $n$ vertices numbered from 1 to $n$ (we call these numbers the indexes of the vertices).\footnote{The assumption that the vertices are numbered from 1 to $n$ is just for the simplicity of presentation. Our algorithms can be adopted to any arbitrary naming of the vertices.} Each machine initially receives a subset of size \Ot{n/\M} of the vertices of \T{} where each vertex object $v$ contains the index of $v$, denoted by \ind{v}, and the index of its parent, denoted by \parent{v}.

In Section~\ref{sec:hash} we show how universal hash functions allow us to efficiently distribute and access objects in our model. In Section~\ref{sec:binary} we introduce an algorithm to convert the given tree into a ``binary-extension'' of it. Finally, in Section~\ref{sec:decompose} we provide an algorithm to decompose the binary extension.

\subsection{Load Balancing via Universal Hash Families}\label{sec:hash}
We start with the definition of universal hash families, that was first proposed by \cite{wegman1981new}.

\begin{definition}\label{def:universalhash}
	A family of hash functions $H = \{h:A \rightarrow B\}$ is $k$-universal if for any hash function $h$ that is chosen uniformly at random from $H$ and  for any $k$ distinct keys $(a_1, a_2, \ldots, a_k) \in A^k$, random variables $h(a_1),  h(a_2), \dots,  h(a_k)$ are independent and uniformly distributed in $B$. A $k$-universal hash function is a hash function that is chosen uniformly at random from a $k$-universal hash family.
\end{definition}

The following lemma shows that we can generate a $(\log \M)$-universal hash function in constant rounds and store it on all the machines. Roughly speaking, the idea is to generate a small set of random coefficients on one machine and share it with all other machines and then use these coefficients to evaluate the hash function for any given input.

\begin{lemma}\label{lem:universalhash}
	For any given $k \in \mathbb{Z}^+$ where $k \leq \poly(n)$, there exists an algorithm that runs in \O{1} rounds of \model{} and generates the same hash function $h: [1, k] \rightarrow \Mset{}$ on all machines where function $h$ is chosen uniformly at random from a $(\log \M)$-universal hash family, and calculating $h(.)$ in a machine takes \O{\log{n}} time and space.
\end{lemma}

We mainly use the universal hash functions to distribute objects into machines. Assuming that an object $a$ has an integer index denoted by \ind{a}, we store object $a$ in machine $h(\ind{a})$ where $h$ is the universal hash function stored on all machines. This also allows other machines to know in which machine an object with a given index is stored.

We claim even if the objects to be distributed by a $(\log \M)$-universal hash function have different sizes (in terms of the memory that they take to be stored), the maximum load on a machine will not be too much with high probability. To that end, we give the following definitions and then formally state the claim.

\begin{definition}
	A set $A = \{a_1, a_2, \ldots, a_k\}$ is a weighted set if each element $a_i \in A$ has an associated non-negative integer weight denoted by $w(a_i)$. For any subset $B$ of $A$, we extend the notion of $w$ such that $w(B) = \sum_{a \in B} w(a)$.
\end{definition}
\begin{definition}\label{def:distributable}
	A weighted set $A = \{a_1, a_2, \ldots, a_k\}$ is a \textit{distributable} set if $w(A)$ is \O{n} and the maximum weight among its elements (i.e., $\max_{a_i \in A}w(a_i)$) is \Ot{n/\M}. \mdcomment{not needed}
\end{definition}

\begin{definition}\label{def:distributer}
	We call a hash function $h: A \rightarrow \Mset$ a \textit{distributer} if $A$ is a weighted set and \Mset{} is the set of all machines. We define the \textit{load} of $h$ for a machine $m$ to be $l_h(m) = \sum_{a \in A: h(a) = m} w(a)$. Moreover we define the maximum load of $h$ to be $\max_{m \in \Mset} l_h(m)$. \mdcomment{not needed}
\end{definition}

\begin{lemma}\label{lem:maxload}
	Let $h: A \rightarrow \Mset$ be a hash function chosen uniformly at random from a $(\log \M)$-universal hash function where $A$ is a distributable set and \Mset{} is the set of machines (i.e., $h$ is a distributer). The maximum load of $h$ is \Ot{n/\M} with probability at least $1-(\frac{2^{1/\delta}e}{\delta\log{n}})^{\delta\log{n} }$ where $m\geq n^\delta$.  
\end{lemma}

The full proof of Lemma~\ref{lem:maxload} is differed to Appendix~\ref{apx:proofs}. The general idea is to first partition the elements of $A$ into \O{n/\M} subsets $S_1, S_2, \ldots, S_{k}$ of size $\O{\M}$ such that the elements are grouped based on their weights (i.e., $S_1$ contains the $\M$ elements of $A$ with the lowest weights, $S_2$ contains the $\M$ elements in $A - S_1$ with the lowest weights, and so on). Then we show that the maximum load of the objects in any set $S_i$ is $\Ot{w_i}$ where $w_i$ is the weight of the object in $S_i$ with the maximum weight. Finally using the fact that the objects are grouped based on their weights, we show that $\sum_{i \in [k]} w_i$, which is the total load is \Ot{n/\M}.

A technique that we use in multiple parts of the algorithm is to distribute the vertices of \T{} among the machines using a $(\log \M)$-universal hash function.

\begin{definition}\label{def:vertex_distribution}
	Let $T$ be a given tree and let $h: \{1, \ldots, |V(T)|\} \to \Mset$ be a mapping of the vertex indexes of $T$ to the machines. We say $T$ is distributed by $h$ if any vertex $v$ of $T$ is stored in machine $h(\ind{v})$.
\end{definition}

As a corollary of Lemma~\ref{lem:universalhash}:

\begin{corollary}\label{cor:distribution}
	One can distribute a tree $T$ by a hash function $h$ that is chosen uniformly at random from a $(\log \M)$-universal hash family in \O{1} rounds in such a way that each machine can evaluate $h$ in \O{\log{n}} time and space. \mdcomment{not needed}
\end{corollary}

\subsection{Conversion of \T{} to a Binary Tree}\label{sec:binary}
The first step towards finding a decomposition of \T{} with the desired properties is to convert it into a binary tree \Tb{} that preserves the important characteristics of the original tree (e.g., the ancestors of each vertex must remain to be its ancestors in the binary tree too).

The definition of an \textit{extension} of a tree is as follows. By the end of this section, we prove it is possible to find a binary extension of \T{} in constant rounds.

\begin{definition}\label{def:extension}
	A rooted tree $\T'{}$ is an extension of a given rooted tree $\T{}$ if  $|V(\T'{})| = \O{|V(\T{})|}$ and there exists a mapping function $f: V(\T{}) \rightarrow V(\T'{})$ such that for any $v \in V(\T)$, $\ind{v} = \ind{f(v)}$ and if $u$ is an ancestor of $v$ in $\T{}$, $f(u)$ is also an ancestor of $f(v)$ in $\T'{}$. \mdcomment{ Suree rooted binary extension}
\end{definition} 

The following lemma proves it is possible to find the degree of all vertices in constant rounds.
\begin{lemma}\label{lem:degree}
	There exists a randomized algorithm to find the degree of each vertex of a given rooted tree $\T$ in \O{1} rounds of \model{} using $\M$ machines, such that with high probability, each machine uses a space of size at most \Ot{n/\M} and runs an algorithm that is linear in its memory size.
\end{lemma}
\begin{proof}
	The first step is to distribute the vertices of \T{} using a $(\log \M)$-universal hash function $h$. By Corollary~\ref{cor:distribution} this takes \O{1} rounds. Define the local degree of a vertex $v$ on a machine $\mu$ to be the the number of children of $v$ that are stored in $\mu$ and denote it by $d_{\mu, v}$. We know there are at most \Ot{n/\M} vertices with a non-zero local degree in each machine. Hence in one round, every machine can calculate and store the local degree of every such vertex. In the communication phase, every machine $\mu$, for any vertex $v$ such that $d_{\mu, v} > 0$, sends $d_{\mu, v}$ to machine $h(v)$.	Then in the next round, for any vertex $v$, machine $h(v)$ will receive the local degree of $v$ on all other machines and can calculate its total degree. 
	
	We claim with high probability, no machine receives more than \Ot{n/\M} data after the communication phase.
	
	Define the weight of a vertex $v$, denoted by $w(v)$, to be $\min\{\M, \deg(v)\}$. The weight of a vertex $v$ is an upper bound on the size of communication that machine $h(v)$ receives for vertex $v$. To see this, observe that a machine sends the local degree of $v$ to $h(v)$ only if this local degree is non-zero, therefore the total communication size for $v$ cannot exceed its degree. On the other hand, the total number of machines is \M{}, hence machine $h(v)$ cannot receive more than \M{} different local degrees for vertex $v$.
	
	We first prove that $V(\T)$ is a distributable set (Definition~\ref{def:distributable}) based on weight function $w$. To see this, note that first, by definition of $w$, the maximum value that $w$ gets is $\M$, which indeed is \Ot{n/\M} (recall that in the model we assumed $m$ is less than the space on each machine and hence $m = \Ot{n/m}$), and second, the total weight of all vertices is \O{n} since $\sum_{v\in V(\T)}\deg(v) = 2n - 2$.
	
	By Lemma~\ref{lem:maxload}, since $V(\T)$ is a distributable set based on $w$ and since $h$ is chosen from a $(\log \M)$-independent hash family, the maximum load of $h$, i.e., the maximum communication size of any machine is \Ot{n/\M} with high probability.
\end{proof}

We are now ready to give a binary extension of \T{}.

\begin{lemma}\label{lemma:binary-tree} \mdcomment{give prover references to this lemma}
	There exists a randomized algorithm  to convert a given rooted tree $\T{}$ to a binary tree $\Tb{}$, an extension of $\T{}$, in \O{1} rounds of \model{} using $\M$ machines, such that with high probability, each machine uses a memory of size at most \Ot{n/\M} and runs an algorithm that is linear in its memory size.
\end{lemma}

\begin{proof}
	Our algorithm for converting \T{} to a binary extension of it done in two phases. In the first phase, we convert \T{} to an extension of it, \Td, with a maximum degree of \Ot{n/\M}. Then in the next phase, we convert \Td{} to a binary extension.
			
	\begin{claim} \label{claim:log-tree}
		There exists a randomized algorithm to convert a given rooted tree $\T{}$ to a tree $\T^d{}$, with maximum degree \Ot{n/\M} that is an extension of $\T{}$ in \O{1} rounds of \model{} such that with high probability, each machine uses a memory of size at most \Ot{n/\M} and runs an algorithm that is linear in its memory size.
	\end{claim}
	
	The detailed proof of Claim~\ref{claim:log-tree} and the pseudo-code to implement it is given in Appendix~\ref{apx:proofs}. Intuitively, after calculating the degree of each vertex, we send the index of all vertices with degree more than \O{n/\M} to all machines (call them high degree vertices). This is possible because there are at most \O{\M} high degree vertices and \M{} is assumed to be less than the space of each machine. To any high degree vertex $v$, in the next step, we add a set of at most \O{n/\M} auxiliary children and set the parent of any previous child of $v$ to be one of these auxiliary vertices that is chosen uniformly at random. The detailed proof of why no machine violates its memory size during the process and how we assign indexes to the auxiliary vertices is differed to Appendix~\ref{apx:proofs}.
	
	The next phase is to convert this bounded degree tree to a binary tree.
		
	\begin{claim}\label{claim:binary-tree}
		There exists a randomized algorithm to convert a tree $\T{}^d$, with maximum degree \Ot{n/\M}, to a binary extension of it $\Tb{}$, in \O{1} rounds of \model{} using $\M$ machines, such that with high probability, each machine uses a memory of size at most \Ot{n/\M} and runs an algorithm that is linear in its memory size.
	\end{claim}
	
	Again, the detailed proof of Claim~\ref{claim:binary-tree} and a pseudo-code for implementing it in the desired setting is given in Appendix~\ref{apx:proofs}. Roughly speaking, since the degree of each vertex is at most \Ot{n/\M}, we can store all the children of any vertex in the same machine (the parent may be stored in another machine). Having all the children of a vertex in the same machine allows us to add auxiliary vertices and locally reduce the number of children of each vertex to at most 2. Figure~\ref{fig:binary} illustrates how we convert the tree to its binary extension. See Appendix~\ref{apx:proofs} for more details.
	\begin{figure}
		\centering
		\includegraphics[scale=0.6]{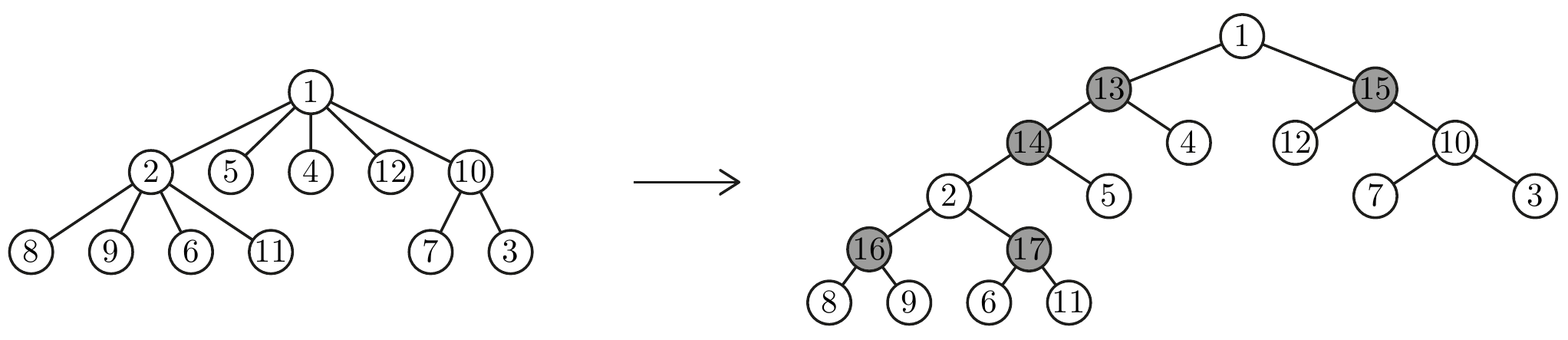}

		\caption{An example of converting a tree to a binary extension of it. The gray nodes denote the added auxiliary vertices and the numbers within the nodes denote their indexes.}
		\label{fig:binary}
	\end{figure}

	By Claim \ref{claim:log-tree} and Claim \ref{claim:binary-tree}, for any given tree, there exists an algorithm to construct a binary extension of it in \O{1} round of \model{}, which with high probability uses \Ot{n/\M} time and space.
\end{proof}

\subsection{Decomposing a Binary Tree}\label{sec:decompose}

We start with the formal definition of decompositions.

\begin{definition}\label{def:decomposition}
	We define a decomposition of a binary tree \Tb{} to be a set of \O{\M} components, where each component contains a subset  of the vertices of \Tb{} that are connected, and each vertex of \Tb{} is in exactly one component. A component $c_i$, in addition to a subset of $V(\Tb{})$ (which we denote by $V(c_i)$) and their ``inner edges", i.e., the edges between the vertices in $V(c_i)$, contains their ``outer edges" (i.e., the edges between a vertex in $V(c_i)$ and a vertex in another component) too. The data stored in each component, including the vertices and their inner and outer edges should be of size up to \Ot{n/\M}. \mdcomment{ SUREE }
\end{definition}

Since the vertices in any component $c_i$ are connected, there is exactly one vertex in $c_i$ that its parent is not in $c_i$, we call this vertex the root vertex of $c_i$ and denote it by \croot{c_i}. We also define the index of a component $c_i$ to be equal to the index of its root; i.e., $\ind{c_i} = \ind{\croot{c_i}}$ (recall that we assume every vertex of \Tb{} has a unique index). Moreover, by contracting a component we mean contracting all of its vertices and only keeping their outer edges to other components.	Since the vertices in the same component have to be connected by definition, the result of contraction will be a rooted tree. We may use this fact and refer to the components as nodes and even say a component $c_i$ is the parent (child, resp.) of $c_j$ if after contraction, $c_i$ is indeed the parent (child, resp.) of $c_j$.

\begin{theorem}\label{thm:decomposebinary}
	There exists a randomized algorithm to find a decomposition of a given binary tree \Tb{} with $n$ vertices using \M{} machines, such that with high probability the algorithm terminates in $\O{\log n}$ rounds of \model{}, and each machine uses a memory of size at most $\Ot{n/\M}$ and runs an algorithm that is linear in its memory size.
\end{theorem}

	Algorithm~\ref{alg:decompose} gives a sequential view of the decomposition algorithm we use. We first prove some properties of the algorithm, and then using these properties, we prove it can actually be implemented in the desired setting.

	\begin{algorithm}[h]
		\caption{The sequential algorithm for finding a valid \model{} decomposition of a given binary tree \Tb{}.}
		\label{alg:decompose}
		\begin{algorithmic}[1]
			\Require A binary tree \Tb{}.
			\State $C \leftarrow \emptyset$ \label{line:cempty} \Comment{$C$ will contain all components.}
			\ForEach{vertex $v$ in \Tb{}}
			\State Insert a component containing only $v$ to $C$.
			\EndFor
			\State $F \leftarrow \emptyset$ \label{line:fempty} \Comment{$F$ will contain the completed components.}
			\While{$|C| > 14m$}\label{line:while}
			\State $S \leftarrow \emptyset$ \label{line:defines} \Comment{$S$ will contain the selected 	components.}
			\ForEach{component $c \in C-F$}\label{line:selects}
			\If{$c$ is the root component}
			\State Add $c$ to $S$.
			\ElsIf{$c$ has exactly two children components}
			\State Add $c$ to $S$. \label{line:twochildren}
			\ElsIf{the parent component of $c$ is in $F$}\label{line:ifparentcompleted}
			\State Add $c$ to $S$. \label{line:parentcompleted}
			\ElsIf{$c$ has exactly one child component}
			\State Add $c$ to $S$ with probability $0.5$. \label{line:onechild} 
			\EndIf
			\EndFor\label{line:selectf}
			\ForEach {component $c \in C-F-S$}\label{line:mergefor}
			\State  $d \gets $ Find the closest ancestor of $c$ that is in $S$.
			\State Merge $c$ into $d$.
			\EndFor \label{line:endmergefor}
			\For{every node $c$ in $C-F$}\label{line:completefor}
			\If{$c$ contains at least $n/\M$ vertices of $\Tb{}$} 
			\State Add $c$ to $F$. \label{line:addtoc}
			\EndIf
			\EndFor\label{line:endcompletefor}
			\EndWhile
		\end{algorithmic}
	\end{algorithm}
	
	The algorithm starts with $n$ components where each component contains only one vertex of \Tb{}. Then it merges them in several iterations until there are only \O{\M} components left. At the end of each iteration, if the size of a component is more than $n/\M$, we mark it as a completed component and never merge it with any other component. This enables us to prove that the components' size does not exceed \Ot{n/\M}. In the merging process, at each iteration, we ``select'' some of the components and merge any unselected component to its closest selected ancestor. We select a component if at least one of the following conditions hold for it.
1. It is the root component.
		2. It has exactly two children.
		3. Its parent component is marked as complete.
	In addition to these conditions, if a component has exactly one child, we select it with an independent probability of $0.5$. Figure~\ref{fig:decompose} illustrates the iterations of Algorithm~\ref{alg:decompose}.
	
	\begin{figure}[h]
		\centering
		\includegraphics[scale=0.7]{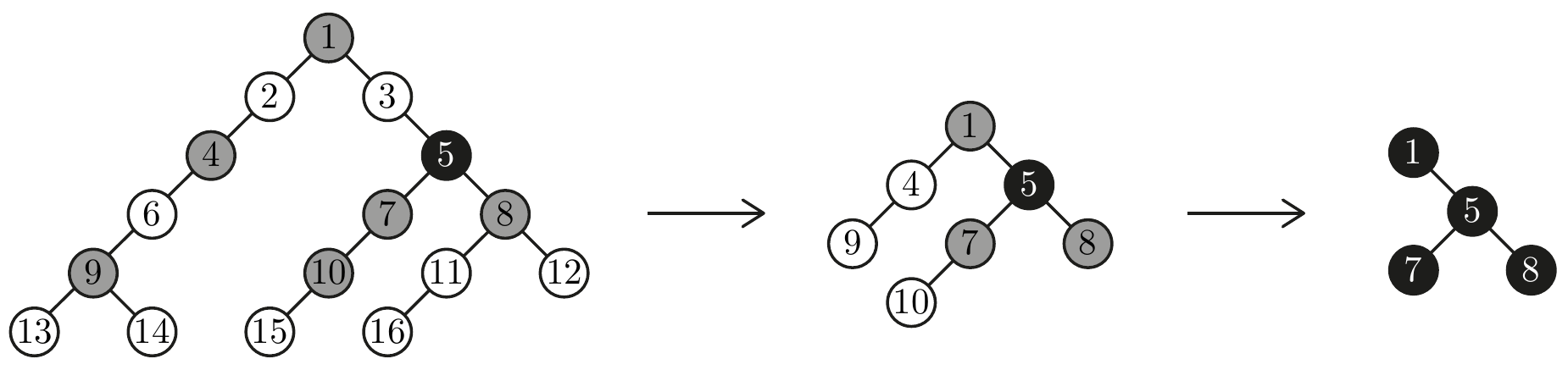}
		\caption{An example of running Algorithm~\ref{alg:decompose}. Each node denotes a component. A gray node denotes a selected component and a black node denotes a completed component.}
		\label{fig:decompose}
	\end{figure}

	The intuition behind the selecting conditions is as follows. The first condition ensures every unselected component has at least one selected ancestor to be merged to. The second condition ensures no component will have more than two children after the merging step. (To see this, let $c_1$ be a component with two children, $c_2$ and $c_3$. Also let $c_3$ have two children, $c_4$ and $c_5$. If only $c_3$ is not selected and is merged to $c_1$, then the resulting component will have three children $c_2$, $c_4$ and $c_5$.) The third condition ensures the path of a component to its closest selected ancestor does not contain any completed component. This condition is important to keep the components connected. Finally, randomly selecting some of the components with exactly one child ensures the size of no component exceeds \Ot{n/\M}. (Otherwise, a long chain of components may be merged to one component in one iteration.) We also prove at most a constant fraction of the components are selected at each iteration. This proves the total number of iterations, which directly impacts the number of rounds in the parallel implementation, is logarithmic in the number of vertices.

	\subsubsection{Properties of Algorithm~\ref{alg:decompose}}
	In this section, we prove some properties of Algorithm~\ref{alg:decompose}. These properties will be useful in designing and analyzing the parallel implementation of the algorithm. 
	
	Many of these properties are defined on the iterations of the while loop at line~\ref{line:while} of Algorithm~\ref{alg:decompose}. Any time we say an iteration of Algorithm~\ref{alg:decompose}, we refer to an iteration of this while loop. We start with the following definition.
	
	\begin{definition}\label{def:algitervars}
		We denote the total number of iterations of of Algorithm~\ref{alg:decompose} by $I$. Moreover, for any $i \in [I]$, we use $C_i$ and $F_i$ to respectively denote the value of variables $C$ (which contains all components) and $F$ (which contains the completed components) at the start of the $i$-th iteration. Analogously, we use $S_i$ to denote the selected components ($S$) at the end of the $i$-th iteration. We also define $T_i$ to be the rooted component tree that is given by contracting all components in $C_i$ (e.g., $T_1 = \Tb{}$).
	\end{definition}
	
	\begin{claim}\label{cl:connected}
		For any $i \in [I-1]$, each component $c \in C_{i+1}$ is obtained by merging a connected subset of the components in $C_{i}$. 
	\end{claim}
	\begin{proof}
		Observe that if in a step a component $v$ is the closest selected ancestor of component $u$, it is also the closest selected ancestor of every other component in the path from $u$ to $v$ (denote the set of these components by $P_u$). We claim if $u$ is merged to $v$, every component in $P_u$ will also be merged to $v$ in the same step. To see this, note that if there exists a component in $P_u$ that is in $F_i$ or $S_i$, the closest selected ancestor of $u$ would not be $v$. Hence all of these components will be merged to $v$ and the resulting component will be obtained by merging a connected subset of the components in $C_{i}$.
	\end{proof}
	
	The following is a rather general fact that is useful in proving some of the properties of Algorithm~\ref{alg:decompose}.
	\begin{fact}\label{fa:leavesnumber}
		For a given binary tree $T$, let $k_i(T)$ denote the number of vertices of $T$ with exactly $i$ children. Then $k_0(T) = k_2(T) + 1$.
	\end{fact}
	
	We claim the components' tree, in all iterations of the while loop, is a binary tree.
	\begin{claim}\label{claim:binary-decomposition}
		For any $i \in [I]$, the component tree $T_i$ is a binary tree.
	\end{claim}

	\begin{proof}
		We use induction to prove this property. We know by Definition~\ref{def:algitervars}, that $T_1 = \Tb{}$ and since \Tb{} is a binary tree, the condition holds for the base case $T_1$.
		
		Assuming that $T_{i-1}$ is a binary tree, we prove $T_i$ is also a binary tree. Fix any arbitrary component $c$ in $C_i$, we prove it cannot have more than two children. Let $U = \{c_1, \ldots, c_k\}$ be the components in $C_{i-1}$ that were merged with each other to create $c$. For any component $c_j \in U$ let \childin{c_j}{U} denote the number of children of $c_j$ that are in $U$ and let \child{c_j} denote the total number of children of $c_j$ (no matter if they are in $U$ or not). For any child of $c$, there is a component in $U$ with a child that is outside of $U$, therefore to prove $c$ cannot have more than two children it suffices to prove
		\begin{equation}\label{eq:2child}
			\sum_{c_j \in U} \big(\child{c_j} - \childin{c_j}{U}\big) \leq 2.
		\end{equation}
		
		We first prove $\sum_{c_j \in U} \child{c_j} \leq |U|+1$. Note that at most one of the components in $U$ has two children in $T_{i-1}$ since all vertices with two children are selected in line~\ref{line:twochildren} of Algorithm~\ref{alg:decompose} and if two components are merged at least one of them is not selected (line~\ref{line:mergefor} of Algorithm~\ref{alg:decompose}). Therefore other than one vertex in $U$ which might have two children in $C_{i-1}$, others have at most one children. Hence $\sum_{c_j \in U} \child{c_j} \leq |U|+1$.
		
		Next, we prove $\sum_{c_j \in U} \childin{c_j}{U} \geq |U|-1$. We know by Claim~\ref{cl:connected} that the components in $U$ are connected to each other. Therefore there must be at least $|U|-1$ edges within them and each such edge indicates a parent-child relation in $U$. Hence $\sum_{c_j \in U} \childin{c_j}{U} \geq |U|-1$.
		
		Combining $\sum_{c_j \in U} \child{c_j} \leq |U|+1$ and $\sum_{c_j \in U} \childin{c_j}{U} \geq |U|-1$ we obtain Equation~\ref{eq:2child} always holds and therefore $T_i$ is a binary tree. 
	\end{proof}
	
	The following property bounds the number of unselected ancestors of a component that is going to be merged to its closest selected ancestor.
	
	\begin{claim}\label{claim:closestancestor}
	With probability at least $1-1/n$, for any given $i$ and $c_i$, where  $i \in [I]$ and $c_i \in C_i - F_i - S_i$, there are at most $2\log{n}$ components in the path from $c_i$ to its closest selected ancestor. 
	\end{claim}
	\begin{proof}
		 We first show with probability $1-1/n^2$, there are at most $2\log{n}$ components in the path from $c$ to its closest selected ancestor. To see this, let $P_c$ denote the set of ancestors of $c$ that have distance at most $2\log{n}$ to $c$. We prove with probability $1 - 1/n^2$ at least one component in $P_c$ is selected. Since, all the children of any component in $F_i$ is in $S_i$, If $P_c \cap F_i \neq \emptyset $, then $P_c \cap S_i \neq \emptyset $. Otherwise, any component in $P_c$ is selected independently with probability at least $\frac{1}{2}$ and the probability that $P_c \cap S_i \neq \emptyset $ is $1-2^{-2\log{n}} = 1 - 1/n^2$.
		 
		In addition we prove that $\Sigma_{i\in [I]} |C_i - F_i - S_i|\leq n$, so that we can use union bound over all the components in  $\cup_{i\in [I]} C_i - F_i - S_i$. To prove this it suffices to prove that  any vertex $v \in V(\T{})$ is at most the root vertex of one of the components in $\cup_{i\in [I]} C_i - F_i - S_i$. To prove this, let $j$ denote the first iteration that $v$ is the root vertex of a component in $C_j - F_j - S_j$. In this iteration we merge this component with one of its ancestors, so it can not  be root vertex of any other component in $\cup_{i\in [I], i>j} C_i - F_i - S_i$.
		 Therefore using union bound over all the components in $\cup_{i\in [I]} C_i - F_i - S_i$, we obtain that with probability at least $1 - 1/n$ for any given $i$ and $c_i$, where  $i \in [I]$ and $c_i \in C_i - F_i - S_i$, there are at most $2\log{n}$ components in the path from $c_i$ to its closest selected ancestor.
	\end{proof}
	
	The following property implies that w.h.p. at each round at most a constant fraction of the incomplete components are selected.
	
	\begin{claim}\label{claim:selectedsize}
		There exists a constant number $0 < c  < 1$ such that for any $i \in [K]$,  with probability at least $1-e^{-n/24}$, $ |S_i| \leq c\cdot|C_i - F_i|.$
	\end{claim}
	\begin{proof}
		Let $K_j(C_i)$ denote the set of components in $C_i$ with exactly $j$ children. Since each component in $K_1(C_i)$ is selected independently with probability $\frac{1}{2}$, by Chernoff bound, in round $i$,
		$$\Pr[|K_1(C_i)\cap S_i| > \frac{3}{4} |K_1(C_i)|] \leq e^{-n/24}.$$
		
		The components that are selected in each round are as follows. With probability at least $1-e^{-n/24}$ we choose less than three forths of the components in $K_1(C_i)$, the root component, all the components in $K_j(C_i)$, and at most two components for any component in $F_i$. Therefore, with probability at least $1-e^{-n/24}$, 
		\begin{equation}\label{eq:selected1}
		|S_i| < 1+ |K_2(C_i)|+ \frac{3}{4}|K_1(C_i)|+ 2|F_i|.
		\end{equation}
		 By Fact \ref{fa:leavesnumber}, $|K_0(C_i)|$= $|K_2(C_i)|+1$, and by Lemma \ref{claim:binary-decomposition}, $T_i$ is a binary tree, so  $C_i= K_0(C_i)+ K_1(C_i)+ K_2(C_i).$ Therefore, 
		 \begin{equation}\label{eq:selected2}
		 1+ |K_2(C_i)|+ \frac{3}{4}|K_1(C_i)| < \frac{1}{2} (|K_2(C_i)|+ |K_0(C_i)|+1) + \frac{3}{4}|K_1(C_i)| \leq \frac{3}{4} |C_i|.
		 \end{equation}
		 By equations \ref{eq:selected1} and \ref{eq:selected2} with probability at least $1-e^{-n/24}$, 
		$	|S_i| \leq \frac{3}{4} |C_i| + 2|F_i|.$ 
			Therefore to complete the proof, it suffices to prove  that there exists a constant number $0 < c  < 1$   such that  equation
		 $$ \frac{3}{4} |C_i| + 2|F_i| \leq c\cdot|C_i - F_i|,$$ which is equivalent to  $$ \frac{2-c}{c-3/4}\leq  \frac{|C_i|}{|F_i|}$$ holds.
		 Since the size of any component in $F_i$ is at least $n/\M$, $|F_i|\leq\M$. In addition, by line 3 of Algorithm \ref{alg:decompose}, $|C_i|> 14m$, therefore $\frac{|C_i|}{|F_i|}\geq 14$. It is easy to see that the equation $$ \frac{2-c}{c-3/4}\leq 14 \leq \frac{|C_i|}{|F_i|}$$
		 holds for $c=5/6$, which is a constant number in (0, 1).
	\end{proof}
	
	The following property is a direct corollary of Claim~\ref{claim:selectedsize}.
	\begin{corollary} \label{cor:log-round}
		With high probability the while loop at line~\ref{line:while} of Algorithm~\ref{alg:decompose}, has at most \O{\log n} iterations.
	\end{corollary}
	
	We are now ready to prove each component has at most \Ot{n/\M} vertices.
	\begin{claim}\label{claim:component-size}
		The size of each components at any iteration of Algorithm~\ref{alg:decompose} is bounded by \Ot{n/\M} with high probability.
	\end{claim}
	\begin{proof}
		By Claim \ref{cl:connected} any component $c_i \in C_{i}$ is obtained by merging a connected subset of the components in $C_{i-1}$. Let $X_i \subset C_{i-1}$  denote this subset, and let $T(X_i)$ denote the component tree given by contracting all components in $X_i$.  We first prove that $|X_i|=\O{\log{n}}$.
		
		By Claim \ref{claim:binary-decomposition}, $T(X_i)$ is a binary tree, and by line 16 of Algorithm \ref{alg:decompose} $X_i- T(X_i).root \subset C_{i-1} - F_{i-1}- S_{i-1}$. Note that, any component in $C_{i-1} - F_{i-1}- S_{i-1}$  has at most 1 other component as a child, and $T(X_i).root$ has at most 2 children. In addition, by Claim \ref{claim:closestancestor}, with high probability for any $c \in X_i- T(X_i).\text{root}$ the distance between $c$ and $T(X_i).\text{root}$ is $2\log{n}$. Therefore, with high probability $|X_i|\leq 4\log{n}+1$, which is \O{\log{n}}.
		
		In addition, note that $X_i \subset C_{i-1} - F_{i-1}$, and any component in $C_{i-1}-F_{i-1}$ has size at most $\O{n/\M}$, hence with high probability $|c_i|=\Sigma_{c\in X_i} |c| = \Ot{n/\M}.$   
	\end{proof}

\section{$(\poly\log)$-Expressible Problems}\label{sec:polylogexpressible}
The goal of this section is to prove and give examples of Theorem~\ref{thm:polylog-expressible} which we restate below.

\xhdr{Theorem~\ref{thm:polylog-expressible}.} For any given $\M$, there exists an algorithm to solve any $(\poly\log)$-expressible problem in \O{\log n} rounds of \model{} using $\M$ machines, such that with high probability each machine uses a space of size \Ot{n/\M} and runs an algorithm that is linear in its input size.

\begin{proof}
Assume we are given a $(\poly\log)$-expressible problem $P$. We first convert the given tree into a binary extension of it using Lemma~\ref{lemma:binary-tree} in \O{1} rounds and then decompose the binary extension using Theorem~\ref{thm:decomposebinary} in \O{\log n} rounds. Note that since any component has at most two children (Claim~\ref{claim:binary-decomposition}), there are at most 2 leaves in any component that is unknown. Therefore we can use the compressor algorithm that is guaranteed to exist since $P$ is assumed to be $(\poly\log)$-expressible and store a data of size at most $\poly\log(n)$ for each component. Then in the next step, since there are at most \Ot{m} components, we can send the partial data of all components into one machine. Then in only one round, we are able to merge all this partial data  since the merger algorithm (that is guaranteed to exist since $P$ is $(\poly\log)$-expressible) takes polylogarithmic time and space to calculate the answer of the dynamic programming for each of the components and it suffices to start from the leave components and compute the value of the dynamic programming for each of them one by one (all of this is done in one step because we have the data of all the components in one machine).
\end{proof}

We first formally prove why Maximum Weighted Matching is a $(\poly\log)$-expressible problem and therefore as a corollary, it can be solved in $\O{\log n}$ rounds.

\newcommand{\Mb}{\ensuremath{M^b}}

\paragraph{Maximum Weighted Matching.} An edge-weighted tree \T{} is given in the input. (Let $w_e$ denote the weight of an edge $e$.) The problem is to find a sub-graph $M$ of \T{} such that the degree of each vertex in $M$ is at most one and the weight of $M$ (which is $\sum_{e \in M} w_e$) is maximized.

At first, we explain the sequential DP that given any binary extension  \Tb{} (see Definition \ref{def:extension}) of an input tree \T{}, finds a maximum matching of \T{} and then prove this DP is $(\poly\log)$-expressible.

Define a subset \Mb{} of the edges of $\Tb{}$ to be an \textit{extended matching} if it has the following properties.
	1.~For any original vertex in \Tb{}, at most one of its edges is in \Mb{}.
	2.~For any auxiliary vertex $v$ of \Tb{}, at most one of the edges between $v$ and its children is in \Mb{}.
	3.~The edge between an auxiliary vertex $v$ and its parent is in \Mb{} if and only if an edge between $v$ and one of its children is in \Mb{}.

Consider an edge $e$ between a vertex $u$ and its parent $v$ in \Tb{}. If $u$ is an auxiliary vertex, we define $w_e$ to be 0 and if $v$ is an original vertex, we define $w_e$ to be the weight of the edge between $v'$ and its parent where $v'$ is the vertex in \T{} that is equivalent to $v$. Furthermore, we define the weight of an extended matching \Mb{} of \Tb{} to be $\sum_{e \in \Mb{}} w_e$.

Let $\Mb{}$ be an extended matching of \Tb{} and let $M$ be a matching of \Tb{}; we say $\Mb{}$ and $M$ are equivalent if for any edge $(u, v) \in M$, every edge between the equivalent vertices of $v$ and $u$ in \Tb{} is in $\Mb{}$ and vice versa. Let $v', u' \in V(\Tb)$ denote the equivalent vertices of $v, u \in V(\T{})$ respectively. If $v$ is the parent of $u$ in \T{}, all the other vertices in the path between $v'$ and $u'$ are auxiliary by definition of binary extensions. This means all the edges in this path, except for the edge between $u'$ and its parent (which is equal to $w_{u, v}$), have weight 0. Therefore the weight of any extended matching of \Tb{} is equal to its equivalent matching of \T{}. Hence, to find the maximum weighted matching of \T{}, one can find the extended matching of \Tb{} with the maximum weight.

Now, we give a sequential DP that finds the extended matching of \Tb{} with the maximum weight. To do so, for any vertex $v$ of \Tb{}, define $C(v)$ to be the maximum weight of an extended matching of the subtree of $v$ where at least one edge of $v$ to its children is part of the extended matching (if no such extended matching exists for $v$ set $C(v)$ to be $-\infty$. Also define $C'(v)$ to be the maximum weight of an extended matching of the subtree of $v$ where no edge of $v$ to its children is part of the extended matching.

The key observation in updating $C(v)$ and $C'(v)$ for a vertex $v$ is that if $v$ is connected to one of its original children $u$, then $u$ should not be connected to its children; if $v$ is

To update $C(v)$ and $C'(v)$, if $v$ has no children, then $C(v) = -\infty$ and  $C'(v) = 0$, also if $v$ has only one child $u$, then $C(v) = w_{(v, u)}$ and $C'(v)=0$. W.l.o.g. we consider the following three cases for the case where $v$ has two children $u_1$ and $u_2$.
\begin{itemize}
	\item If $u_1$ and $u_2$ are both original vertices, then
$$C(v) = \max{\{w_{v, u_1} + C'(u_1) + C(u_2), w_{v, u_2} + C'(u_2) + C(u_1)\}},$$ $$C'(v) = \max{\{C(u_1), C'(u_1)\}} + \max{\{C(u_2), C'(u_2)\}}.$$
	\item If $u_1$ and $u_2$ are both auxiliary vertices, then
	$$C(v) = \max{\{w_{v, u_1} + C(u_1) + C'(u_2), w_{v, u_2} + C(u_2) + C'(u_1)\}},$$ $$C'(v) = \max{\{C'(u_1), C'(u_2)\}}.$$
	\item If $u_1$ is auxiliary and $u_2$ is an original vertex, then 
	$$C(v) = \max{\Big\{w_{v, u_1} + C(u_1) + \max{\{C'(u_2), C(u_2)\}}, w_{v, u_2}, + C'(u_2) + C'(u_1)\}\Big\}},$$ $$C'(v) = C'(u_1) + \max{\{C(u_2), C'(u_2)\}}.$$
\end{itemize}

Now we prove that the $(\poly\log)$-expressiveness property holds for the proposed DP. Consider a subtree $T$ of \Tb{} and let $u_1$ and $u_2$ \sbcomment{expressiveness is defined for constant leaves not only two} denote the leaves of $T$ with the unknown DP values. We apply the proposed DP rules and calculate the DP values as functions of $C(u_1)$, $C'(u_2)$, $C(u_2)$, and $C'(u_2)$ instead of plain numbers. We claim that for any vertex $v$ of $T$ there exist functions $f_v, f'_v: \{0, 1\}^4 \to \mathbb{Z}$ such that $$C(v) = \max_{a_i \in \{0, 1\}}{\{a_0C(u_1) + a_1C'(u_1) + a_2C(u_2) + a_3C'(u_2) + f_v(a_0, a_1, a_2, a_3)\}},$$
and 
$$C'(v) = \max_{a_i \in \{0, 1\}}{\{a_0C(u_1) + a_1C'(u_1) + a_2C(u_2) + a_3C'(u_2) + f'_v(a_0, a_1, a_2, a_3)\}}.$$
Therefore it suffices to store the values of functions $f'_r$ and $f_r$ for the root vertex $r$ of $T$ to be able to evaluate $C(v)$ and $C'(v)$. Since $f_r$ and $f'_r$ take only 16 different input combinations, there are only \O{1} different output values to be stored.

\begin{corollary}
	For any given $\M$, there exists an algorithm to solve the Maximum Weight Matching problem on trees in \O{\log n} rounds of \model{} using $\M$ machines, such that with high probability each machine uses a space of size \Ot{n/\M} and runs an algorithm that is linear in its input size.
\end{corollary}

\paragraph{Other problems.} The dynamic programming solutions of Maximum Independent Set, Minimum Vertex Cover, Longest Path and many other problems are similar to the proposed DP for Maximum Weighted Matching. For example, for the Maximum Independent Set problem, we only need to keep two dynamic data, whether the root vertex of a subtree is part of the independent set in its subtree or not and for each of these cases what is the size of the maximum independent set in the subtree. Roughly speaking, the only major difference in the DP is that we have to change $\max$ to sum.
\sbcomment{More intuitions for other problems.}

\begin{corollary}
	For any given $\M$, there exists algorithms to solve the Maximum Independent Set, Minimum Vertex Cover, Longest Path and Dominating Set problems on trees in \O{\log n} rounds of \model{} using $\M$ machines, such that with high probability each machine uses a space of size \Ot{n/\M} and runs an algorithm that is linear in its input size.
\end{corollary}

\section{Solving Linear-Expressible Problems}\label{sec:linearexpressible}

\subsection{The Splittability Property}\label{sec:splittability}
The main goal of this section is to define a property called ``splittability'' for linear-expressible problems, that indeed holds for the aforementioned linear-expressible problems and prove, by using a total memory of \Ot{n^{4/3}}, we are able to parallelize linear-expressible problems that are splittable in \O{\log n} rounds. The main reason behind the need to define this extra property for linear-expressible problems, is that in contrast to $(\poly\log)$-expressible problems, the partial data of the subtrees in linear-expressible problems might have a relatively large size, hence when merging two subtrees, we are not able to access their partial data in one machine. As a result, we need to be able to distribute the merging step among different machines. It is worth mentioning that we prove the total memory of \Ot{n^{4/3}} is in some sense tight, unless we limit the number of machines to be too small, which defeats the purpose of MapReduce, which is to have maximum possible parallelization.

\newcommand{\subunifier}{\ensuremath{\mathcal{S}}}
\newcommand{\unifier}{\ensuremath{\mathcal{U}}}
\renewcommand{\vector}[1]{\ensuremath{\vec{\compressor(#1)}}}

Intuitively, a problem instance $P$ is splittable if it is possible to represent it with a set of $n$ {\em objects} such that the output of the problem depends only on the pair-wise relations of these objects. In the DP context, the splittability property is useful for linear-expressible problems. The splittability property ensures that these partial data can be effectively merged using the splitting technique.

\begin{definition}[Splittability] \label{def:Splittability}
	Consider a linear-expressible problem $P$ and a binary adapted dynamic program \sol{P} for it. We know since $P$ is a linear-expressible problem, for any subtree $T$ of \Tb{} with $n_T$ vertices, a compressor function returns a partial data of size at most \Ot{n_T} if at most a constant number of the leaves of $T$ are unknown. Denote this partial data by $\compressor(T)$. 
	We say problem $P$ is splittable, if for any given connected subtree $T'$ of \Tb{}, the partial data $\compressor(T')$ can be represented as a vector $\vector{T'}$ of at most $\Ot{n_{T'}}$ elements such that for any connected subtrees $T_1$ and $T_2$ of \Tb{} that are connected with an edge (with $T$ being $T_1 + T_2$) the following condition holds. There exist two algorithms \subunifier{} (sub unifier) and \unifier{} (unifier) such that for any consecutive partitioning $(P_1, P_2, \ldots, P_{k_1})$ of \vector{T_1} and any consecutive partitioning $(Q_1, Q_2, \ldots, Q_{k_2})$ of \vector{T_2},
	\begin{enumerate}
		\item Algorithm $\subunifier(P_i, Q_j)$ returns a vector of size at most $|P_i| + |Q_j|$ in time $\poly(|P_i|+|Q_j|)$ such that each element of this vector ``affects'' at most one element of \vector{T}.
		\item Given the elements in $\subunifier(P_i, Q_j)$ that ``affect'' an element $a$ in \vector{T} for every  $1 \leq i \leq k_1$ and $1 \leq j \leq k_2$, $\unifier$ computes the value of element $a$ using linear space and polynomial time in input size.
	\end{enumerate}
\end{definition}

The linear-expressible problems such as minimum bisection, $k$-spanning tree, etc., that we have mentioned before are all indeed splittable. For the proof and concrete examples of unifier and sub unifier algorithms see Section~\ref{app:linearexpressible}.

%
%

\subsection{Solving Splittable Linear-Expressible Problems}

We show with a slightly more total space (the space on each machine is still sublinear) it is possible to design parallel algorithms in \O{\log n} rounds for linear-expressible problems if they are splittable.

\xhdr{Theorem~\ref{thm:linear-expressible}.} For any given $\M$, there exists an algorithm to solve any linear-expressible problem defined on trees that is also splittable, in \O{\log n} rounds of \model{} using $\M$ machines, such that with high probability each machine uses a space of size \Ot{n^{4/3}/\M}, and runs an algorithm that is polynomial in  its input size.
\vspace{0.2cm}

As a corollary of Theorem~\ref{thm:linear-expressible}, we give logarithmic round algorithms for minimum bisection, and $k$-spanning tree problems using an overall space of $\Ot{n^{1+\eps}}$ for some constant $\eps < 1$.

Our algorithm, similar to the algorithm given for $(\poly\log)$-expressible problems starts with decomposing a binary extension of the input tree. However, note that for linear-expressible problems, we cannot store the data of all components in one machine since the space limit will be violated. (The size of the total data would be roughly the same as the size of input.) Therefore we merge the components in multiple rounds. Unfortunately, the components tree could be unbalanced. To overcome this difficulty, we first partition the components tree. This partitioning of the components tree, specifies which components at which rounds, should be merged together. The merging process cannot be done in one machine either. Mainly because the partial data of even one component could be very large. Using the splittability property we distribute the data of the components among the machines and merge them in parallel.

We first prove the partitioning lemma and then prove Theorem~\ref{thm:linear-expressible}.
\begin{definition}[border vertices]
	For any subtree $T$ of a given tree \T{}, the border vertices are the set of vertices in $V(T)$ that are connected with an edge of tree \T{} to a vertex in $V(\T{})-V(T)$.
\end{definition}

\begin{lemma}\label{lem:partitioning}
	There exists a linear time algorithm that given a binary tree $\T{}$ with $n_\T{}$ vertices, and a set of border vertices $S_T\in V(\T{})$ with size at most 3, finds a partitioning $P_\T{} = \langle H_1, H_2, \ldots, H_c \rangle$ of $\T{}$, such that $|V(H_i)| \leq \frac{2n_\T{}}{3}$ for any $i \leq c$, and the total number of border vertices of any $H_i$ based on $(P_\T{}, S_\T{})$ is at most 3.
\end{lemma}

We defer the proof of Lemma~\ref{lem:partitioning} to the appendix. We are now ready to prove Theorem~\ref{thm:linear-expressible}.

\begin{proof}[Proof of Theorem~\ref{thm:linear-expressible}]
Let \T{} denote the input tree. Since any linear-expressible problem is binary adaptable, there exists a dynamic program \sol{P} where for the instance $P_\T$ of $P$ with \T{} as its input tree and for any binary extension \Tb{} of \T{}, the output of \sol{P}(\Tb) is a solution of $P_\T$. We convert \T{} to its binary extension \Tb{} in $\O{1}$  round of \model{} such that with high probability each machine uses space \Ot{n/\M} (see Lemma \ref{lemma:binary-tree}). Since the output of \sol{P} is a solution of $P_\T$, it suffices to solve \sol{P}.

The first stage in this algorithm is decomposing tree \Tb{}, which  by Theorem \ref{thm:decomposebinary} is possible in $\O{\log{n}}$ round of \model{} such that with high probability each machine uses $\Ot{n/\M}$ time and space. Let $C(\Tb{})$ denote the set of components in decomposition of $\T^b{}$ by Definition \ref{def:decomposition}, and let $\T{}^c$ denote the rooted component tree that is given by contracting all components in $C(\Tb{})$. Roughly speaking, we partition $\T{}^c$ to at most three subtrees $T{}_1, T{}_2$, and $T{}_3$, such that $\max_{i=1}^{3}{|\T{}_i|} \leq \frac{2}{3}|\T{}^c|$, and $\max_{i=1}^{3}{|B_i|}  \leq 3$, where $B_i$ denotes the set of border vertices in $\T{}_i$. We find the partial dynamic data for these partitions recursively and we merge them to find the dynamic data for $\T{}^c$. Algorithm \ref{alg:bisection} is the overview of this algorithm.
	
\begin{algorithm}
\caption{} \label{alg:bisection}
\begin{algorithmic}[1]
\Procedure{SolveLinearExpressible}{$T$}
	\State Decompose \Tb{} and let $C(\Tb{})$ denote the set of components in decomposition of $\T^b{}$.
	\State Let $\T{}^c$ be the component tree that is given by contracting all components in $C(\Tb{})$ .
	\State \Call{GetDynamicData}{$\T^c{}$, $\emptyset$}
	\EndProcedure
	\Procedure{GetDynamicData}{$\T{}$, $B$} \Comment{$B$ is the set of border vertices in $V(\T{})$}
	\If{$|\T{}|= 1$}	
		\State Based on $B$, compute and return the dynamic data for $T$.
		\Else
			\State Partition $\T{}$ to partitions $T{}_1, T{}_2$, and $T{}_3$, such that (1) $\max_{i=1}^{3}{|\T{}_i|} \leq \frac{2}{3}|\T{}|$, (2) $\max_{i=1}^{3}{|B_i|}  \leq 3$ where $B_i$ denotes the set of border vertices in $\T{}_i$.
			\For{every $i$ in $\{1,2,3\}$}
			\State $D_i \leftarrow$ \Call{GetDynamicData}{$\T_i{}$, $B_i$}
			\EndFor
			\State Compute the dynamic data for \T{}, by merging $D_1, D_2,$ and $D_3$, and return it.
	\EndIf
	\EndProcedure
\end{algorithmic}
\end{algorithm}
 
\paragraph{Partitioning stage.} The second stage of our algorithm  partitions $\T^c{}$ using the algorithm given in Lemma~\ref{lem:partitioning}. Let $n_c$ denote the number of vertices of tree $\T^c{}$. Starting from $\T^c{}$, in each step, we partition the remaining subtrees by removing one edge from each and repeat this until we have $n_c$ disjoint vertices. Since each sub-tree is split  into sub-trees of a constant fraction size of it, this stage halts after $\O{\log{n_c}}$ steps, so the overall running time of this stage is $\Ot{n_c}\leq \Ot{\M}$. At the end of this stage, for each edge $e$ of $\T^c{}$, we store in which step it was deleted and use this data  to design an algorithm to merge the dynamic  data of the sub-trees and finally find an optimal solution for $\T^d{}$. More precisely, after this stage, the edges of $\T^c{}$ are partitioned into subsets $E_1, E_2, \ldots, E_{2S}$ where $S=\O{\log{n}}$ is the number of steps this stage takes, and for any $i\in[S]$, $E_{2i-1}$ and $E_{2i}$ are respectively the set of edges deleted by Algorithm~\ref{alg:First_cut} and Algorithm~\ref{alg:Second_cut} in the $(S-i+1)$-th step of this stage. We store $E_1, E_2, \ldots, E_{2S}$ in any machine in $\Mset$ without violating the memory limits since the total number of edges of $\T^c{}$ is at most \Ot{\M}

The third stage of the algorithm is the merging stage that runs exactly in the reverse order of the previous stage. Starting from single vertices of $\T^c{}$, which are equivalent to components in $C(\Tb{})$, in the first round, the algorithm continues to merge the  dynamic data of two subtrees if there exists an edge $e \in E_1$ that connects them. Note that the size of any component in $C(\Tb{})$ is at most $\Ot{n/\M}$, so it is possible to find its partial dynamic data in one machine.

Let vector $D(T)$ denote the dynamic data (It is the partial dynamic data if there is any unknown data in $T$.) for any subtree $T$. Since the maximum number of border vertices for any partition is a constant, by linear-expressible (Definition \ref{def:Expressiveness}) we can store the dynamic data for any partition in linear space to the size of the partition. Moreover by the same definition, if $T=T_1+T_2$ it is possible to construct the dynamic data for $T$ by merging the dynamic data for $T_1$ and $T_2$. Note that to update the dynamic data for any sub-tree $T$, we have to update $O(|V(T)|)$ different values, and to update each of them, it is possible that we need more than\O{n/\M} space. This unfortunately cannot be done in one machine since it violates the memory limit of $O(n^{1+\eps})$ on each machine. To handle this issue we use the following method to merge the dynamic data of two components. 

\paragraph{Merging stage.} This stage of algorithm takes $2R$ steps. Intuitively, for any $1\leq r\leq 2R$ and for any edge in $E_r$ that is between two partitions $T^r_i$ and $T^r_j$, in the first round of step $r$ we assign any machine a portion of dynamic data of $T^r_i$ and $T^r_j$. This machine is responsible to merge these two portions using the function \subunifier{} (sub unifier)in Definition \ref{def:Splittability}. Let $T^{r+1}$ be the subtree which is the result of merging subtrees $T^r_i$ and $T^r_j$ in step $r$, and let  $D(T^r_i)[b]$ denote the $b$-th element of vector $D(T^r_i)$. More precisely, for any two positive integers $k_1, k_2 \leq \lceil\sqrt{\M}\rceil$ dynamic data $D(T^r_i)[(k_1-1) \alpha_i :k_1 \alpha_i]$ and $D(T^r_j)[(k_2-1) \alpha_j :k_2 \alpha_j]$ is assigned to machine $\mu$ with index $k_1\sqrt{\M}+k_2$, where for any $i$, $\alpha_i=|T^r_i|/\sqrt{\M}$. (For any subtree $T$ and any two positive integers $a, b$ such that $a< b$ $D(T)[a:b]$ denotes the sub-vector $(D(T)[a],\dots, D(T)[b])$.) This machine is responsible to merge the given portion of dynamic data and generate $D^\mu(T)$, such that $|D^\mu(T)| = \O{D(T^r_i)+D(T^r_j)}$, and any element in $D^\mu(T)$ affects at most one element in $D(T)$. This is possible by function \subunifier{}. Therefore to compute $D(T)[a]$  for any valid $a$ in one round of \model{} it suffices if for every $\mu\in\Mset$ we have the element in $D^\mu[T]$ that affects $D(T)[a]$ (if there exists any)  in the same machine, and then use the function \unifier{} (unifier). To achieve this for any $T^{r+1}$ in the set of partitions of $\Tb{}$ in round $r$ and for any $a\leq |V(T^{r+1})|$ we assign each $D(T^{r+1})[a]$ a unique index in $[|V(\Tb)|]$ generated based on $T^{r+1}$ and $a$. Using this index and a hash function $h$ chosen uniformly at random from a family of $\log{n}-$universal hash functions, we distribute pairs of $(T^{r+1},a)$ over all $\M$ machines, and any machine $\mu$ sends at most one element of $D^\mu_c(T^{r+1})$  that affects $D(T^{r+1})[a]$ to the machine $\mu'$ iff pair $(T^{r+1}, a)$ is assigned to machine $\mu'$ using hash function $h$. Therefore, all the data needed to compute $D(T^{r+1})[a]$, is gathered in one machine, and in the second round of $r$-th step, we run function \unifier{} on all the machines to compute the dynamic data of $T^{r+1}$.

	  	\paragraph{Time and space analysis.} By Lemma \ref{lemma:binary-tree} and Theorem \ref{thm:decomposebinary} it is possible to convert any given tree to a binary extension of it, and decompose the binary extension in $\Ot{\log{n}}$ rounds of \model{} using $\M$ machines, such that with high probability, each machine uses a memory of size at most $\Ot{n^{1+\eps}/\M}$ and runs an algorithm that is linear in its memory size. Since the size of the component tree $\T{}^c$ is $\O{\M}$, and the partitioning stage takes $S=\O{\log{n}}$ steps, by \ref{lem:partitioning} it is possible to do the partitioning stage of the algorithm using  $\Ot{n^{1+\eps}/\M}$ time and space in $\O{1}$ round. (Note that $\Ot{\M} \leq \Ot{n^{1+\eps}/\M}$.)  In addition, the merging stage of the algorithm takes $S=\O{\log{n}}$ round, and the space needed on each machine is $\Ot{n^{1+\eps}/\M}$. Since in round $r$ each vertex in \Tb{} is at most in one partition, the total size of the partitions in one round is $\O{n}$. Note that, the dynamic data of each partition is linear to its size. Therefore when in the first round of any step $r\in [2R]$ of merging we assign each machine $\O{1/\sqrt{\M}}$ portion of the dynamic data of each partition, the overall data assigned to each machine is $\O{n/\sqrt{\M}}$. We set $\eps=\frac{1}{3}$. Since $\M < n^{(1+\eps)/2}$ it is easy to see that $\O{n/\sqrt{\M}} \leq \O{n^{1+\epsilon}/\M}$, so computing function \subunifier in each round takes $\O{n^{1+\epsilon}/\M}$ space and $\poly(n^{1+\epsilon}/\M$) time in any machine. Moreover, with high probability the data assigned to each machine in the second round of any step $r$ is $\Ot{n^{1+\epsilon}/\M}$ since we distribute any pair of partition $T$ in step $r$ and $a\leq |T|$ using a hash function $h$ chosen uniformly at random from a family of $\log{n}-$universal hash functions. The maximum data that any machine receives for any pair $(T, a)$ that is assigned to it by hash function $h$ is at most \O{\M} and the overall data assigned to all the machines is $\O{n\sqrt{\M}}$. Note that, $\O{n\sqrt{\M}} \leq \O{n^{1+\eps}}$. Therefore by Lemma \ref{lem:maxload} with high probability the maximum load over all the machines is $\Ot{n^{1+\epsilon}/\M}$, so each machine takes $\Ot{n^{1+\epsilon}/\M}$ space and poly($n^{1+\epsilon}/\M$) time to compute unifier for all the pairs $(T,a)$ that are assigned to it in any round.
\end{proof}

\section{A Concrete Example: Minimum Bisection on Trees}\label{app:linearexpressible}
	We continue the discussion about linear-expressible problems by giving a detailed proof that the minimum bisection problem can be solved in logarithmic rounds.

\begin{theorem}
For any given $m$, there exists an algorithm to solve the minimum bisection problem in \O{\log{n}} rounds of \model{} using $m$ machines, such that with high probability each machine uses a space of size \Ot{n^{4/3}/m}.
\end{theorem}

\begin{proof}

	While converting $\T{}$ to binary version of it, for any auxiliary vertex $a$ we set the weight of the edge between $a$ and its parent to be $\infty$. Note that in any optimal solution for the new tree, the infinity weights enforce all auxiliary vertices to be in the same bisection set that their parents are, thus any solution for the original tree will have the same cost for the binary tree as long as we make sure the auxiliary vertices are in the same set that their parents are (otherwise the cost will be $\infty$). Recall that in the bisection's solution the number of vertices of the two sets should be equal, hence we have to keep track of the auxiliary vertices that we added to the tree and make sure we do not count them.

	The dynamic data $D_c(T_i)[b]$, for any partition (sub-tree) $T_i$, and any possible coloring $c$ of its border vertices into red and blue sets, and any number $b \leq |V(T_i)|$, is the minimum cost of partitioning $T_i$ such that the partitions match what is specified in $c$ and there are exactly $b$ non-auxiliary vertices in the blue set. (Recall that the color of auxiliary vertices will have to be the same as their parents or the cost would be $\infty$). Note that we fix the exact coloring of border vertices by $c$ to make it possible to merge the partial dynamic data for two components. Two components are only connected through their border vertices and if we know how the border vertices are colored, we will know if in a solution we have to add the weight of the edges between border vertices to the cost or not. Note that there are at most $4$ border vertices in any partition so there are constant number of possible cases for $c$. This implies that bisection is linear-expressible.

	   If we define the unifier (\unifier) and the sub unifier (\subunifier) functions for the dynamic data stored for each partition we obtain that bisection is also splittable (Definition \ref{def:Splittability}).
	Let $T$ be the result of merging $T_i$ and $T_j$, and let set $D(T')[b]$ denote the set of all $D_c(T')[b]$ for any possible coloring $c$ of border vertices in $V(T')$, and any partition $T'$. The function $\subunifier$ for any consecutive portion of $D(T_i)$ from element $D(T_i)[a_i]$ to $D_c(T_i)[b_i]$ denoted by $D(T_i)[a_i:b_i]$, and any consecutive portion of $D(T_j)$ denoted by $D(T_j)[a_j:b_j]$ is as follows: for any pair $k_i$ and $k_j$ such that $a_i\leq k_i\leq b_i$ , $a_j\leq k_j\leq b_j$, the value of pair ($D(T_i)[k_i]$, $D(T_j)[k_j]$) is only related to $D(T)[k_i+k_j]$. Therefore, for any possible coloring $c$ of border vertices of $T$, and any number $x$ such that $a_i+a_j\leq x\leq  b_i+b_j$,
	$\subunifier$ outputs a value for $D_c(T)[x]$ which is the maximum  value for $D_c(T)[x]$ based on the input given to this function. It is easy to see that $\unifier$ for any $c$ is the maximum of all its inputs that match with the color of border vertices $c$, which is a function polynomial to its input.
 
	  	At the end, $\max_c D_c(\Td)[\frac{n'}{2}]$, where $n'$ is the number of non-auxiliary vertices of $\T$, is the solution we return.
	  
Therefore, by Theorem \ref{thm:linear-expressible} since bisection is linear- and splittable, there exists an algorithm to solve bisection in \Ot{\log{n}} round of \model{} using $\M$ machines, such that with high probability each machine uses \Ot{n^{1+\eps}/\Mset} space, and takes \O{n^2/\M} time in each round of the algorithm, for $\eps=\frac{1}{3}$.
\end{proof}

\paragraph{Extensions.} It could be shown that with very similar techniques, the ``$k$-maximum spanning tree" problem on a tree could also be solved. In this problem a weighted tree is given and the goal is to find the maximum weighted connected subgraph of the given tree with exactly $k$ vertices. When we convert this tree to a binary tree we only need to make sure that we mark auxiliary vertices such that we do not disconnect them from their parents in the algorithm. Moreover, instead of considering all the possible coloring of the border vertices in bisection, in the dynamic for $k$-maximum spanning tree, we need to consider all the cases for the connectivity of the border vertices with the maximum weighted tree in the partition. 

Furthermore, we expect that this idea can be extended to several other problems,
including facility location problem, minimum diameter, $k$ partitioning, $k$-median, and $k$-center, because they have very similar dynamic programming solutions
when the input graph is a tree.

\subparagraph*{Acknowledgements.}
	The authors thank Silvio Lattanzi for fruitful discussions and his helpful comments.
	
\bibliography{ref}

\begin{thebibliography}{10}

\bibitem{abrahamson1989simple}
Karl Abrahamson, Norm Dadoun, David~G. Kirkpatrick, and T~Przytycka.
\newblock A simple parallel tree contraction algorithm.
\newblock {\em Journal of Algorithms}, 10(2):287--302, 1989.

\bibitem{ahn2015access}
Kook~Jin Ahn and Sudipto Guha.
\newblock Access to data and number of iterations: Dual primal algorithms for
  maximum matching under resource constraints.
\newblock In {\em Proceedings of the 27th ACM symposium on Parallelism in
  Algorithms and Architectures}, pages 202--211. ACM, 2015.

\bibitem{DBLP:conf/soda/AhnGM12}
Kook~Jin Ahn, Sudipto Guha, and Andrew McGregor.
\newblock Analyzing graph structure via linear measurements.
\newblock In {\em Proceedings of the Twenty-Third Annual {ACM-SIAM} Symposium
  on Discrete Algorithms}, pages 459--467, 2012.

\bibitem{andoni2014parallel}
Alexandr Andoni, Aleksandar Nikolov, Krzysztof Onak, and Grigory Yaroslavtsev.
\newblock Parallel algorithms for geometric graph problems.
\newblock In {\em Proceedings of the 46th Annual ACM Symposium on Theory of
  Computing}, pages 574--583. ACM, 2014.

\bibitem{bahmani2012densest}
Bahman Bahmani, Ravi Kumar, and Sergei Vassilvitskii.
\newblock Densest subgraph in streaming and mapreduce.
\newblock {\em Proceedings of the VLDB Endowment}, 5(5):454--465, 2012.

\bibitem{bahmani2012scalable}
Bahman Bahmani, Benjamin Moseley, Andrea Vattani, Ravi Kumar, and Sergei
  Vassilvitskii.
\newblock Scalable k-means++.
\newblock {\em Proceedings of the VLDB Endowment}, 5(7):622--633, 2012.

\bibitem{DBLP:conf/nips/BateniBDHKLM17}
MohammadHossein Bateni, Soheil Behnezhad, Mahsa Derakhshan, MohammadTaghi
  Hajiaghayi, Raimondas Kiveris, Silvio Lattanzi, and Vahab~S. Mirrokni.
\newblock Affinity clustering: Hierarchical clustering at scale.
\newblock In {\em Advances in Neural Information Processing Systems 30: Annual
  Conference on Neural Information Processing Systems 2017, 4-9 December 2017,
  Long Beach, CA, {USA}}, pages 6867--6877, 2017.

\bibitem{bateni_et_al:LIPIcs:2018:9166}
MohammadHossein Bateni, Soheil Behnezhad, Mahsa Derakhshan, MohammadTaghi
  Hajiaghayi, and Vahab Mirrokni.
\newblock {Brief Announcement: MapReduce Algorithms for Massive Trees}.
\newblock In Ioannis Chatzigiannakis, Christos Kaklamanis, D{\'a}niel Marx, and
  Donald Sannella, editors, {\em 45th International Colloquium on Automata,
  Languages, and Programming (ICALP 2018)}, volume 107 of {\em Leibniz
  International Proceedings in Informatics (LIPIcs)}, pages 162:1--162:4,
  Dagstuhl, Germany, 2018. Schloss Dagstuhl--Leibniz-Zentrum fuer Informatik.

\bibitem{bateni2014distributed}
MohammadHossein Bateni, Aditya Bhaskara, Silvio Lattanzi, and Vahab Mirrokni.
\newblock Distributed balanced clustering via mapping coresets.
\newblock In {\em Advances in Neural Information Processing Systems}, pages
  2591--2599, 2014.

\bibitem{beame2013communication}
Paul Beame, Paraschos Koutris, and Dan Suciu.
\newblock Communication steps for parallel query processing.
\newblock In {\em Proceedings of the 32nd ACM SIGMOD-SIGACT-SIGAI symposium on
  Principles of database systems}, pages 273--284. ACM, 2013.

\bibitem{DBLP:conf/spaa/BehnezhadDETY17}
Soheil Behnezhad, Mahsa Derakhshan, Hossein Esfandiari, Elif Tan, and Hadi
  Yami.
\newblock Brief announcement: Graph matching in massive datasets.
\newblock In {\em Proceedings of the 29th {ACM} Symposium on Parallelism in
  Algorithms and Architectures, {SPAA} 2017, Washington DC, USA, July 24-26,
  2017}, pages 133--136, 2017.

\bibitem{DBLP:journals/corr/abs-1802-10297}
Soheil Behnezhad, Mahsa Derakhshan, and MohammadTaghi Hajiaghayi.
\newblock Brief announcement: Semi-mapreduce meets congested clique.
\newblock {\em CoRR}, abs/1802.10297, 2018.

\bibitem{DBLP:journals/corr/abs-1807-06701}
Soheil Behnezhad, Mahsa Derakhshan, MohammadTaghi Hajiaghayi, and Richard~M.
  Karp.
\newblock Massively parallel symmetry breaking on sparse graphs: {MIS} and
  maximal matching.
\newblock {\em CoRR}, abs/1807.06701, 2018.

\bibitem{mistreesmpc}
Sebastian Brandt, Manuela Fischer, and Jara Uitto.
\newblock {Breaking the Linear-Memory Barrier in {MPC:} Fast {MIS} on Trees
  with n\({}^{\mbox{{\(\epsilon\)}}}\) Memory per Machine}.
\newblock {\em CoRR}, abs/1802.06748, 2018.

\bibitem{concurrent}
Sebastian Brandt, Manuela Fischer, and Jara Uitto.
\newblock {Matching and MIS for Uniformly Sparse Graphs in the Low-Memory MPC
  Model}.
\newblock {\em CoRR}, abs/1807.05374, 2018.

\bibitem{DBLP:conf/kdd/ChierichettiDK14}
Flavio Chierichetti, Nilesh~N. Dalvi, and Ravi Kumar.
\newblock Correlation clustering in mapreduce.
\newblock In {\em The 20th {ACM} {SIGKDD} International Conference on Knowledge
  Discovery and Data Mining}, pages 641--650, 2014.

\bibitem{DBLP:conf/soda/ChitnisCEHMMV16}
Rajesh Chitnis, Graham Cormode, Hossein Esfandiari, MohammadTaghi Hajiaghayi,
  Andrew McGregor, Morteza Monemizadeh, and Sofya Vorotnikova.
\newblock Kernelization via sampling with applications to finding matchings and
  related problems in dynamic graph streams.
\newblock In {\em Proceedings of the Twenty-Seventh Annual {ACM-SIAM} Symposium
  on Discrete Algorithms}, pages 1326--1344, 2016.

\bibitem{DBLP:conf/soda/ChitnisCHM15}
Rajesh~Hemant Chitnis, Graham Cormode, Mohammad~Taghi Hajiaghayi, and Morteza
  Monemizadeh.
\newblock Parameterized streaming: Maximal matching and vertex cover.
\newblock In {\em Proceedings of the Twenty-Sixth Annual {ACM-SIAM} Symposium
  on Discrete Algorithms}, pages 1234--1251, 2015.

\bibitem{dean2008mapreduce}
Jeffrey Dean and Sanjay Ghemawat.
\newblock {MapReduce}: Simplified data processing on large clusters.
\newblock {\em Communications of the ACM}, 51(1):107--113, 2008.

\bibitem{ene2015random}
Alina Ene and Huy~L Nguyen.
\newblock Random coordinate descent methods for minimizing decomposable
  submodular functions.
\newblock In {\em ICML}, pages 787--795, 2015.

\bibitem{DBLP:conf/soda/EsfandiariHLMO15}
Hossein Esfandiari, Mohammad~Taghi Hajiaghayi, Vahid Liaghat, Morteza
  Monemizadeh, and Krzysztof Onak.
\newblock Streaming algorithms for estimating the matching size in planar
  graphs and beyond.
\newblock In {\em Proceedings of the Twenty-Sixth Annual {ACM-SIAM} Symposium
  on Discrete Algorithms}, pages 1217--1233, 2015.

\bibitem{feldman2010distributing}
Jon Feldman, S~Muthukrishnan, Anastasios Sidiropoulos, Cliff Stein, and Zoya
  Svitkina.
\newblock On distributing symmetric streaming computations.
\newblock {\em ACM Transactions on Algorithms (TALG)}, 6(4):66, 2010.

\bibitem{gazit1988optimal}
Hillel Gazit, Gary~L Miller, and Shang-Hua Teng.
\newblock Optimal tree contraction in the erew model.
\newblock In {\em Concurrent Computations}, pages 139--156. Springer, 1988.

\bibitem{DBLP:journals/corr/abs-1807-06251}
Mohsen Ghaffari and Jara Uitto.
\newblock Sparsifying distributed algorithms with ramifications in massively
  parallel computation and centralized local computation.
\newblock {\em CoRR}, abs/1807.06251, 2018.

\bibitem{goel2012complexity}
Ashish Goel and Kamesh Munagala.
\newblock Complexity measures for map-reduce, and comparison to parallel
  computing.
\newblock {\em arXiv preprint arXiv:1211.6526}, 2012.

\bibitem{goodrich2011sorting}
Michael~T Goodrich, Nodari Sitchinava, and Qin Zhang.
\newblock Sorting, searching, and simulation in the mapreduce framework.
\newblock In {\em International Symposium on Algorithms and Computation}, pages
  374--383. Springer, 2011.

\bibitem{guha2003clustering}
Sudipto Guha, Adam Meyerson, Nina Mishra, Rajeev Motwani, and Liadan
  O'Callaghan.
\newblock Clustering data streams: Theory and practice.
\newblock {\em IEEE transactions on knowledge and data engineering},
  15(3):515--528, 2003.

\bibitem{im2015brief}
Sungjin Im and Benjamin Moseley.
\newblock Brief announcement: Fast and better distributed mapreduce algorithms
  for k-center clustering.
\newblock In {\em Proceedings of the 27th ACM symposium on Parallelism in
  Algorithms and Architectures}, pages 65--67. ACM, 2015.

\bibitem{im2017efficient}
Sungjin Im, Benjamin Moseley, and Xiaorui Sun.
\newblock Efficient massively parallel methods for dynamic programming.
\newblock In {\em Proceedings of the 46th Annual ACM Symposium on Theory of
  Computing}. ACM, 2017.

\bibitem{kane2010optimal}
Daniel~M Kane, Jelani Nelson, and David~P Woodruff.
\newblock An optimal algorithm for the distinct elements problem.
\newblock In {\em Proceedings of the twenty-ninth ACM SIGMOD-SIGACT-SIGART
  symposium on Principles of database systems}, pages 41--52. ACM, 2010.

\bibitem{karloff2010model}
Howard Karloff, Siddharth Suri, and Sergei Vassilvitskii.
\newblock A model of computation for mapreduce.
\newblock In {\em Proceedings of the twenty-first annual ACM-SIAM symposium on
  Discrete Algorithms}, pages 938--948. Society for Industrial and Applied
  Mathematics, 2010.

\bibitem{kumar2015fast}
Ravi Kumar, Benjamin Moseley, Sergei Vassilvitskii, and Andrea Vattani.
\newblock Fast greedy algorithms in mapreduce and streaming.
\newblock {\em ACM Transactions on Parallel Computing}, 2(3):14, 2015.

\bibitem{DBLP:conf/spaa/LattanziMSV11}
Silvio Lattanzi, Benjamin Moseley, Siddharth Suri, and Sergei Vassilvitskii.
\newblock Filtering: a method for solving graph problems in mapreduce.
\newblock In {\em Proceedings of the 23rd Annual {ACM} Symposium on Parallelism
  in Algorithms and Architectures}, pages 85--94, 2011.

\bibitem{DBLP:conf/kdd/LucierOS15}
Brendan Lucier, Joel Oren, and Yaron Singer.
\newblock Influence at scale: Distributed computation of complex contagion in
  networks.
\newblock In {\em Proceedings of the 21th {ACM} {SIGKDD} International
  Conference on Knowledge Discovery and Data Mining}, pages 735--744, 2015.

\bibitem{DBLP:conf/stoc/MirrokniZ15}
Vahab~S. Mirrokni and Morteza Zadimoghaddam.
\newblock Randomized composable core-sets for distributed submodular
  maximization.
\newblock In {\em Proceedings of the Forty-Seventh Annual {ACM} on Symposium on
  Theory of Computing}, pages 153--162, 2015.

\bibitem{mirzasoleiman2013distributed}
Baharan Mirzasoleiman, Amin Karbasi, Rik Sarkar, and Andreas Krause.
\newblock Distributed submodular maximization: Identifying representative
  elements in massive data.
\newblock In {\em Advances in Neural Information Processing Systems}, pages
  2049--2057, 2013.

\bibitem{DBLP:journals/corr/abs-1807-08745}
Krzysztof Onak.
\newblock Round compression for parallel graph algorithms in strongly sublinear
  space.
\newblock {\em CoRR}, abs/1807.08745, 2018.

\bibitem{pietracaprina2012space}
Andrea Pietracaprina, Geppino Pucci, Matteo Riondato, Francesco Silvestri, and
  Eli Upfal.
\newblock Space-round tradeoffs for mapreduce computations.
\newblock In {\em Proceedings of the 26th ACM international conference on
  Supercomputing}, pages 235--244. ACM, 2012.

\bibitem{reid1993list}
Margaret Reid-Miller, Gary~L Miller, and Francesmary Modugno.
\newblock List ranking and parallel tree contraction.
\newblock {\em Synthesis of Parallel Algorithms}, pages 115--194, 1993.

\bibitem{roughgarden2016shuffles}
Tim Roughgarden, Sergei Vassilvitskii, and Joshua~R Wang.
\newblock Shuffles and circuits:(on lower bounds for modern parallel
  computation).
\newblock In {\em Proceedings of the 28th ACM Symposium on Parallelism in
  Algorithms and Architectures}, pages 1--12. ACM, 2016.

\bibitem{schmidt1995chernoff}
Jeanette~P Schmidt, Alan Siegel, and Aravind Srinivasan.
\newblock Chernoff--hoeffding bounds for applications with limited
  independence.
\newblock {\em SIAM Journal on Discrete Mathematics}, 8(2):223--250, 1995.

\bibitem{vadhan2012pseudorandomness}
Salil~P Vadhan et~al.
\newblock Pseudorandomness.
\newblock {\em Foundations and Trends{\textregistered} in Theoretical Computer
  Science}, 7(1--3):1--336, 2012.

\bibitem{wegman1981new}
Mark~N Wegman and J~Lawrence Carter.
\newblock New hash functions and their use in authentication and set equality.
\newblock {\em Journal of computer and system sciences}, 22(3):265--279, 1981.

\bibitem{hadoop}
Tom White.
\newblock {\em Hadoop: The Definitive Guide}.
\newblock O'Reilly Media, Inc., 2012.

\bibitem{spark}
Matei Zaharia, Mosharaf Chowdhury, Michael~J. Franklin, Scott Shenker, and Ion
  Stoica.
\newblock Spark: Cluster computing with working sets.
\newblock In {\em Proceedings of the 2Nd USENIX Conference on Hot Topics in
  Cloud Computing}, pages 10--10, 2010.

\end{thebibliography}
\clearpage

\appendix

	
	\section{Omitted Proofs}\label{apx:proofs}
	
\subsection*{Proof of Lemma~\ref{lem:universalhash}}

\begin{proof}
	We use the following proposition by \cite{vadhan2012pseudorandomness}.
	\begin{proposition}[Corollary 3.34 in  \cite{vadhan2012pseudorandomness}]\label{pro:k-universal}
		For  any $d, p, k \in \mathbb{Z}^+$, there is a family of k-wise independent hash functions $H = \{h : [d] \rightarrow \{0, 1,  \dots, 2^p-1\} \}$, from which choosing a random function $h$ takes \mbox{$k\cdot\max\{\log{d},p\}$} random bits, and evaluating any function $h\in H$ takes time $\poly(\log{d},p,k)$.
	\end{proposition}
	
	Recall that we assume the number of machines, \M{}, is a power of two in the model. Therefore we can use Proposition~\ref{pro:k-universal} to choose the hash function $h: [k] \to \Mset{}$ with only $\max{\{\log k, \log \M\}}$ random bits which we can generate on one machine and share with all other machines. The total communication of this machine would be $\O{\M (\log k + \log \M)}$ which is $\Ot{n/\M}$ since $\M = \O{\sqrt{n}}$ and $k$ is $\poly(n)$.
\end{proof}

\subsection*{Proof of Lemma~\ref{lem:maxload}}

To prove Lemma~\ref{lem:maxload} we use a known result on the balls and bins problem by~\cite{schmidt1995chernoff}. We start with the definition of the balls and bins problem.

\begin{definition} [Balls and bins]
	In an instance of balls and bins  problem there are $n$ balls and $n$ bins. Each ball is placed into one of the bins that is chosen uniformly at random.   Let random variable $X_i$ with range $\{1,..,n\}$ denote the bin that the $i$-th ball is assigned to. In this problem, the load of an arbitrary bin $b\in[n]$  is  the number of balls assigned to this bin, which is equal to $\Sigma_{i=1}^{n} X_i=b$.
\end{definition}

The following claim is based on the result of \cite{schmidt1995chernoff}.

\begin{claim}  \label{claim:balls-logwise}
	Let $l_i$ denote load of the $i$-th bin in an instance of balls and bins problem in which random variables $X_1, \dots,  X_n$ are ($\log{n}$)-wise independent, then  $\Pr{[\max_i^{n}l_i > 2\log{n}]} \leq  (\frac{2 e}{\log{n}})^{\log{n} } $.	
\end{claim}
\begin{proof}
	By Theorem  6-(\RN{5}) in \cite{schmidt1995chernoff} we infer that, given an instance of balls and bins problem in which random variables $X_1, \dots,  X_n$ are ($\log{n}$)-wise independent, for any $r$ and $\eta$ such that $r\geq 1+\eta+\log{n}$, and $k \geq \eta$ The following equation holds:
	$$\Pr{[l_i > r ]} \leq \frac{e^{\eta} }{(1+\eta)^{1+\eta}} $$. 
	
	Using this result, by setting $\eta= \log{n} -1$ and $r=2\log{n}$ in the mentioned instance of balls and bins problem, for any $i\in [n]$,  $\Pr{[l_i > 2\log{n}]} \leq  (\frac{e}{\log{n}})^{\log{n} } $. Given this fact, by taking union bound over all $n$ bins we achieve the following result:
	$$ \Pr{[\max_i^{n}l_i > 2\log{n}]} \leq  n \cdot \Pr{[l_i > 2\log{n}]= 2^{\log n}\cdot (\frac{e}{\log{n}})^{\log{n}} =(\frac{2 e}{\log{n}})^{\log{n} } }.\qedhere$$
\end{proof}

We are now ready to prove Lemma~\ref{lem:maxload}.

\begin{proof}[Proof of Lemma~\ref{lem:maxload}]
	\sbcomment{This could be moved to the main part}
	Our starting point is Lemma~\ref{claim:balls-logwise}. We represent the elements of $A$ as balls and the machines as bins. However, there are two main problems. First, $|A|$ may be much larger than the number of machines (\M{}), while Lemma~\ref{claim:balls-logwise} expects the number of balls to be equal to the number of bins. Second, the balls are unweighted in Lemma~\ref{claim:balls-logwise} whereas the elements of A are weighted.
	
	To resolve the above issues, we first partition the elements of $A$ into $\lceil n/\M\rceil$ ($=\setsize$) subsets $S_1, S_2, \ldots, S_{\setsize}$ of size $\M$ (except the last subset $S_{\setsize}$  which might have fewer elements if $n/\M$ is not an integer) such that the elements are grouped based on their weights. More precisely, $S_1$ contains the $\M$ elements of $A$ with the lowest weights, $S_2$ contains the $\M$ elements in $A - S_1$ with the lowest weights, and so on.
	
	\begin{claim}\label{claim:maxloadset}
		For $\M=n^\delta$, with probability at least $1-(\frac{2^{1/\delta}e}{\delta\log{n}})^{\delta\log{n} }$  there is no set $S_i$ where $h$ maps more than \O{\log\M} elements of $S_i$ to the same machine.
	\end{claim}
	\begin{proof}
	Since $h$ is chosen from a $(\log \M)$-universal hash function, and since each $S_i$ has at most $\M$ elements, the machines that the elements in $S_i$ are assigned to are $(\log \M)$-wise independent. Therefore, by fixing a set $S_i$, we can use Lemma~\ref{claim:balls-logwise} to prove at most \O{\log \M} elements of $S_i$ are assigned to the same machine with probability $1 - (2e/\log\M)^{\log\M}$. By taking union bound over the failure probability of all sets we obtain that with probability at least 
	$$1 - \frac{n}{m}(\frac{2e}{\log\M})^{\log\M} = 1-n(\frac{e}{\log{m}})^{\log{m} }= 1-(\frac{2^{1/\delta}e}{\delta\log{n}})^{\delta\log{n} }$$ 
	there is no set $S_i$ where $h$ maps more than \O{\log\M} elements of $S_i$ to the same machine.
	\end{proof}
	
	Let $w_i = \max_{a \in S_i}w(a)$ denote the weight of the element in $S_i$ with the maximum weight. Also let $l^i_h(\mu)$ denote the load of $h$ on machine $\mu$ from the elements in set $S_i$ (i.e., $l^i_h(\mu) = \sum_{a \in S_i: h(a) = \mu}w(a)$). By Claim~\ref{claim:maxloadset}, with high probability, from any set $S_i$, at most \O{\log\M} elements are mapped to the same machine. Also since the maximum weight of the elements in $S_i$ is $w_i$ we know with high probability that $l^i_h(\mu) \leq \O{\log\M\cdot w_i}$ for any $i$. We know by Definition~\ref{def:distributer} that $l_h(\mu) = \sum_{i=1}^{\setsize}l^i_h(\mu)$, therefore 
	\begin{equation}
	l_h(\mu) = l^1_h(\mu) + \ldots + l^{\setsize}_h(\mu) \leq \O{\log\M\cdot(w_1 + \ldots w_{ \setsize})}.
	\end{equation}
	
	 Therefore the maximum load of $h$ is \Ot{w_1 + w_2 + \ldots + w_{\setsize}}. To complete our proof, it suffices to prove $w_1 + w_2 + \ldots + w_{\setsize} = O(n/\M)$. To see this, note that $w(a) \geq w_{i-1}$ for any $a$ in $S_i$ or otherwise $a$ would have been in $S_{i-1}$ instead of $S_i$. Therefore,
	\begin{equation}
	\sum_{a \in S_i} w(a) \geq |S_i| \cdot w_{i-1} = \M \cdot w_{i-1} \qquad \forall i: 1 \leq i < \setsize,
	\end{equation}
	which means,
	\begin{equation}
	\M(w_1 + w_2 + \ldots + w_{\setsize-1}) \leq \sum_{i=2}^{\setsize}\sum_{a \in S_i} w(a).
	\end{equation}
	On the other hand since the total weight of the elements in $A$ is \O{n} by Definition~\ref{def:distributable}, we obtain
	\begin{equation}
	\M(w_1 + w_2 + \ldots + w_{\setsize-1}) \leq \O{n},
	\end{equation}
	and therefore $\sum_{i=1}^{\setsize-1} w_i \leq \O{n/\M}$. Finally since by Definition~\ref{def:distributable} we know the maximum weight in $A$, and consequently $w_{\setsize}$ is not more than $n/\M$, we prove $\sum_{i=1}^{\setsize} w_i \leq \O{n/\M}$ as desired.
\end{proof}

\subsection*{Proof of Lemma~\ref{lem:partitioning}}

\begin{proof}
	The main idea of the algorithm is to first partition $\T{}$ into two sub-trees $\T{}_1$ and $\T{}_2$ such that each of them has at most $\frac{2n_\T{}}{3}$ vertices. Unfortunately, in case all inherited border vertices of $\T{}$ lie within one of the sub-trees, we may end up having 4 border vertices in one of the sub-trees since the partitioning step creates one new border vertex in each of the sub-trees. To handle this case, we again partition the sub-tree with 4 border vertices into two sub-trees such that they will not have more than 3 border vertices at the end.

	Let $D[v]$ \sbcomment{fix notation} denote the number of descendants of vertex $v$ (including itself), and let $B[v]$ denote the number of border vertices among these descendants. It is easy to see that there exists linear time bottom-up dynamic algorithms to calculate these two vectors.
	
	In the first step, we use Algorithm \ref{alg:First_cut} to find an edge and claim removing it partitions $\T{}$ into two sub-trees $\T{}_1$ and $\T{}_2$, where $\max(|V(\T{}_1)|, |V(\T{}_2)|) \leq \frac{2n_\T{}}{3}$.
	
	\begin{algorithm} 
	\caption{} \label{alg:First_cut}
	\begin{algorithmic}[1]
	\Procedure{FirstCut}{$\T{}$}
	\State $D \gets \Call{Descendants}{\T{}}$
	\For {any vertex $v$ of $T$}
		\State $C_D[v] \gets$ the child $u$ of $v$ with maximum $D[u]$
	\EndFor
	\State $n_\T{} \gets D[\T{}.root]$
	\State $v \gets \T{}.root$
	\While{$D[C_D[v]] > \frac{2n_\T{}}{3} $} \label{line:if_first}
		\State $v \gets C_D[v]$
	\EndWhile
	\State \Return edge $(v, C_D[v])$
	\EndProcedure
	\end{algorithmic}
	\end{algorithm}
	
	To prove the returned edge $e = (v, u)$ is a desired one, it suffices to prove $\frac{n_\T{}}{3} \leq D[u]\leq \frac{2n_\T{}}{3}$ since this yields that the size of the other partition is also less than or equal to $\frac{2n_\T{}}{3}$. The condition at Line~\ref{line:if_first} of Algorithm~\ref{alg:First_cut} ensures: (i) $D[u] \leq \frac{2n_\T{}}{3}$, and (ii) $D[v] > \frac{2n_\T{}}{3}$. The former inequality is a direct result of the condition; the latter also holds since otherwise $v$ (the parent of $u$) would not have been updated. Using the second inequality and since $u = C_D[v]$ (i.e., the child of $v$ with maximum number of descendants) we can obtain $D[u] \geq \frac{D[v]-1}{2} \geq \frac{n_\T{}}{3}$ as desired.
	
	After partitioning $\T{}$ we may end up having one partition with more than 3 border vertices (exactly 4) since removing an edge can add up to one new border vertex to each partition. In this case, let $\T{}_1$ denote the sub-tree with $4$ border vertices. We run Algorithm~\ref{alg:Second_cut} to partition $\T{}_1$ into two sub-trees with less than 4 border-vertices.
\begin{algorithm} \mdcomment{revise these algorithms}
\caption{} \label{alg:Second_cut}
\begin{algorithmic}[1]
\Procedure{SecondCut}{$\T{}$}
	\State $B \gets \Call{BorderVertices}{\T{}}$
	\For {any vertex $v$ of $\T{}$}
		\State $C_B[v] \gets$ the child $u$ of $v$ with maximum $B[u]$ 
	\EndFor
	\State $v \gets \T{}.root$
	\While{$B[C_B[v]] > 2$} 
		\State $v \gets C_B[v]$
	\EndWhile
	\State \Return $edge(v, C_B[v])$
	\EndProcedure
\end{algorithmic}
\end{algorithm}
\end{proof}

\subsection*{Proof of Claim \ref{claim:log-tree}}
\begin{proof}
	Algorithm \ref{alg:Tree-log} is the implementation of such an algorithm in \model{}.	Before running this algorithm, we distribute all the vertices in $V(\T{})$ over $\Mset$ machines independently and uniformly at random. To construct $\T^d{}$ we first add $|V(\T{})|$   vertices to it with the same set of indexes as the vertices in $V(\T{})$. These vertices do not have any parent yet, so we assign parents to them later in the algorithm. We will also add some auxiliary vertices to this tree.
		
	Set $\Delta$ to be $\lceil \frac{n}{m}\rceil$.  For any vertex $v\in \T{}$, let $C_v$ denote the set of children of $v$  in tree $\T{}$, and let $v^d$ denote the vertex equivalent  to  $v$ in $\T^d{}$, which has the same index as $v$. If $\Delta \geq |C_v|$ for any vertex $u\in C_v$, we set the parent of vertex $u^d$ as $v^d$. Otherwise, if  $\Delta< |C_v|$, we add a set of $ \frac{|C_v|}{\Delta}$ auxiliary  vertices to $V(\T^d{})$, denoted by \ $A_v$, such that $v^d$ is the parent of all the vertices in  $A_v$, and for any vertex $u\in C_v$, we choose one of the auxiliary vertices in $A_v$ independently and uniformly at random as the parent of the vertex $u^d$.  This guarantees that, for any two vertex $u$ and $v$ that $v$ is an ancestor of $u$ in \T{}, $v^d$ is also an ancestor of $u^d$. In addition $|V(\T^d{})|=\O{|V(\T{})|} $, so $\T^d{}$ is an extension of $T{}$.
	Note that $\frac{|C_v|}{\Delta}\leq m\leq \O{n/\M}= \O{\Delta\log{n}}$ , therefore for any vertex $v$ of $T{}$, the degree  of $v^d$ is at most  $\O{\Delta\log{n}}$. 	Moreover, for any two vertices $u$ and $v$ such that $|C_v|> \Delta$ and $u\in A_v$, it is easy to see that by Claim~\ref{claim:balls-logwise} with high probability the maximum number of vertices from $C_v$ that choose $u$ as the parent for their equivalent vertex in $T^d{}$, is at most  $\O{\Delta\log{n}}$ since any vertex in $C_v$ choses a vertex in $A_v$ independently and uniformly at random. Therefore, the maximum degree in $\T^d{}$ is $\O{\Delta\log{n}}$. 
		
	To implement this algorithm in \model{}, each machine in \Mset{} needs to have the list of all the vertices in \T{} with degree more than $\Delta$ alongside with the number of children each one's equivalent vertex  will have in tree $T^d{}$ \sbcomment{give intuition for why these are needed?}. Let $K_\mu$ \sbcomment{I use $\mu$ as a superscript, better to be consistent. I can't change mine to subscript cause it is used for another attribute.} denote the set of such vertices of \T{} in any machine $\mu \in \Mset$, and let $S_\mu$ be a dictionary with the set of keys $K_\mu$, such that $S_\mu(v)$ is the number of children that vertex $v^d$ will have in tree $\T^d$ for any vertex $v\in K_\mu$. We assume that each machine has the degree of all the vertices that are stored in it. By Lemma \ref{lem:degree} we can achieve this in $\O{1}$ rounds of \model{}. In the first round of the algorithm each machine constructs and sends this dictionary to all machines. In the next round, each machine construct a dictionary $S$ which is the union of all the dictionaries received in round 1, and its own dictionary. (Since the dictionaries have distinct set of keys the union of them is well-defined.)  For any vertex $v \in V(\T{})$ we  construct $v^d$ in the same machine as $v$, and using $S$, we set its parent. In addition, if the degree of $v$ in tree \T{} is more than $\Delta$, we add $\frac{S(v)}{\Delta}$ auxiliary vertices as the children of $v^d$ in this machine. It is easy to see that the data in dictionary $S$, suffices to assign a unique index to any auxiliary vertex. After this algorithm terminates each vertex of $V(\T^d)$ is stored in exactly one machine.
  
  	To prove that during the algorithm, with high probability no machine violates the memory limit of $\Ot{n/\M}$, it suffices to prove that with high probability the number of vertices with degree greater than $\Delta$ in any machine $\mu \in \Mset$ is $\O{\log{n}}$, since it yields that the size of the message each machine sends in round 1 of the algorithm is $\O{\log{n}}$, and hence we have $\Ot{\M} \leq  \Ot{n/\M}$, the overall message passing for each machine does not violate the memory limit. Moreover, it indicates that the total number of vertices added to each machine is at most $ \O{\log{n}} \cdot {\frac{n}{\Delta}} \leq \Ot{n/\M}$. To prove this claim, we use two facts. First, before running the algorithm we uniformly distribute the vertices over all the machines, and second, the total number of the vertices with degree at least $\Delta$ in a tree is at most $\O{\frac{n}{\Delta } }\leq \O{\M}$. So, by Claim \ref{claim:balls-logwise} with high probability there are at most $\O{\log{n}}$  vertices with degree at least $\Delta$ in any machine $\mu \in \Mset$.
\end{proof}
  
  \begin{algorithm}[h]
  	\caption{An algorithm to convert a given tree \T{} to a bounded degree tree $\T^d{}$, an extension of \T{}}
  	\label{alg:Tree-log}
  	\begin{algorithmic}[1]
  		
  		\Require{given a number $\Delta$  this algorithm converts a tree \T{} whose vertices are  distributed  over $\M$ machines to an extension of it, $\T^d{}$ with maximum degree $\O{\Delta\log{n}}$}  		
  		\State Let $\mu$ denote the machine in which this instance of algorithm is running, and let $V_\mu$ denote the set of vertices of $\T{}$  that are stored in this machine.
  		\Round{1}
  		\State Let $S_\mu$ be a dictionary with the set of keys $K_\mu=\{\ind{v} |\deg(v)> \Delta, v\in V_\mu\}$.
  		\State For any $v$, such that $\ind{v}\in K_\mu$, $S_\mu(\ind{v}) = \lceil\frac{|\deg(v)|}{\Delta}\rceil$. \mdcomment{deg(v) ?? children(v)}
  		\State Send $S_\mu$ and $K_\mu$ to all the other machines.
  		\EndRound
  		\Round{2}
  		\State $K \leftarrow \cup_{i \in \Mset} K_i$
  		\State  $S \leftarrow \cup_{i \in \Mset} S_i$ \Comment{Since $K_i$s are distinct, for any $i\in K$ S[i] has exactly one value}
  		
  		\State Let $D$ be a dictionary on the same set of keys as $S$, while for any  $i\in K$, $D[i]=n+\Sigma_{j\in S, j<i}S[j]$.
  		\For {every $v$ in $V_\mu$}
  		\If{$\parent{v} \in K $}
  		\State	Choose $x$ uniformly at random from $[S(\ind{v})]$.
  		\State   $\parent{v}\leftarrow D[\parent{v}]+ x$ 	\Comment{$D[\parent{v}]+ x$ is a unique index in the range $ [n, 2n]$ assigned to the $x$-th child of vertex v in $\T^d{}$}
  		\EndIf
  		\EndFor
  		\For{every $v$ in $K_\mu$}
  		\For{every $j$ in $[S[v]]$}
  		\State Add Vertex$(v, D[v]+j)$ to $V_\mu$ .
  		\EndFor
  		\EndFor
  		\EndRound
  	\end{algorithmic}
  \end{algorithm}

\subsection*{Proof of Claim~\ref{claim:binary-tree}}
\begin{proof}
	\sbcomment{Reminder for myself to have a pass on this proof.}
	Algorithm \ref{alg:binary} is the implementation of an algorithm to convert $\T{}^d$ to $\Tb{}$ in \model{}. To construct $\Tb{}$ we first add $|V(\T{}^d)|$  vertices to it with the same set of indexes as the vertices in $V(\T{}^d)$. These vertices do not have any parent yet, so we assign parents to them later in the algorithm. We will also add some auxiliary vertices to this tree.
	Let $C_v$ denote the set of children of vertex $v$, for any $v\in V(\T^d{})$, and let $v^b$ denote its equivalent vertex in $\Tb{}$ . The overall idea in this algorithm is as follows: for any vertex with more than two children we construct a binary tree $t_v$ with root $v^b$ such that for any $u\in C_v$, $u^b$ is a leave in this tree, and $|V(t_v)|= 2|C_v|-1$. Then we  add the vertices of $t_v$ excluding the root  to binary tree  $\Tb{}$ with the same parent they have in $t_v$.  In addition, for any $v$ with less than 3 children we set the parent of its children's equivalent vertices in $\Tb{}$ as $v^b$.
		
	To implement this algorithm in \model{}, we need to have all the children of any vertex in one machine. 
	To achieve this, we first distribute all the vertices in $V(\T{}^d)$ over $\M$  machines using a hash function $h: [|V(\Tb{})|]\rightarrow[\M]$ which is chosen uniformly at random from a family of $\log{n}$-universal hash functions. By Corollary \ref{cor:distribution}  it is possible to distribute $V(\T{}^d)$   by hash function $h$ in \O{1} round of \model{} , and using  $\O{\log \M}$ time and space on each machine in such a way that each machine can evaluate $h$ in \O{\log \M} time and space. Let $V_\mu$ denote the set of vertices that is assigned to an arbitrary machine $\mu$ by this distribution. For any vertex $v \in V_\mu$ let $p_v$ denote $v$'s parent, and let $\mu_v$ be the machine that  $p_v \in V_{\mu_v}$. Since evaluating $h$ in each machine takes \O{\log \M} time and space, It is possible to find $\mu_v$ for all $v \in V_\mu$ using  \Ot{n/\M} time and space in machine $\mu$. We send any vertex $v$ as a message to the machine $\mu_v$. In this way all the children of any vertex are sent to the same machine. In addition since  $|C_v|\leq \ \Ot{n/\M}$, by Lemma \ref{lem:maxload} with high probability the total size of the messages received by any machine $\mu$ which is equal to $\Sigma_{v in V_\mu} |C_v|$ is  \Ot{n/\M}. 
	
	 For any vertex $v\in V(\T^d{})$ where $|C_v|>2$, we add some auxiliary vertices to \Tb{}, which does not have any equivalent vertex in $V(\T^d{})$, so we need to assign an index to them. To make sure that these assigned indexes are unique, any machine $\mu$ computes the number of all the auxiliary vertices generated in this machine, which is equal to$\Sigma_{v\in V_\mu} \max(|C_v|-2,0)$, and send it to all the other machines with index greater than $\ind{\mu}$. It is easy to see that using this received data each machine is able to generate unique indexes for any auxiliary vertex in it.	
	Figure \ref{fig:binary} represents how this algorithm works. The gray vertices in this figure are the auxiliary vertices added to the tree.
\end{proof}

\begin{algorithm}[h]
	\caption{An algorithm to convert a given bounded degree tree $\T^d{}$ to a binary extension of it, $\T^d{}$}
	\label{alg:binary}
	\begin{algorithmic}[1]
\Require{In the beginning of the  algorithm vertices of the tree $\T^d{}$ with maximum degree $\O{\Delta\log{n}}$ are distributed over $\M$ machines using a hash function $h$. This algorithm converts $\T^d{}$ to a binary tree $\T^b{}$, which is an extension of $\T^d{}$.}
\label{alg:Tree-log)}
\State Let $\mu$ denote the machine in which this instance of algorithm is running, and let $V_\mu$ denote the set of vertices of $\T^d{}$  that are stored in this machine. 
	\Round{1}
		\For {every $v$ in $V_m$}
	\State $\mu_d=h(\parent{v})$
	\State Send $(\parent{v}, v)$ to machine $\mu_d$. \Comment{Send it even if $\mu_d = \mu$.}
	\State Delete v from $V_\mu$.
	\EndFor
	\EndRound
		
	\Round{2}
		\State Let $R$ denote the set of all the messages received in the previous round.
		\State $C_v \leftarrow\{u| (v, u) \in R\}$ \Comment $C_v$ is the set of all the children of vertex $v$
		\State Send $\Sigma_{v\in V_\mu} \max(|C_v|-2,0)$ to any machine $\mu' $ such that $\ind{\mu'} > \ind{\mu}$.	
		\EndRound
		\Round{3}
		\State	Let s denote the summation of all the messages received in the previous round.
		\State $d \leftarrow s+1$
		\For {every $v$ in $V_\mu$}
		\State Construct a binary tree $b_v$ with root $v$, leaves $C_v$, and $|C_v|-2$ inner vertices indexed by $d, \dots, d+|C_v|-2$. \Comment{Any vertex in a tree that is not a root, nor a  leave is considered as an inner vertex. If $|C_v|\leq 2$ there is no inner vertex.}
		\State $d \leftarrow d+ \max(0,|C_v|-2)$
		\For{every vertex $u$ in $b_v$  excluding $v$}
		\State Let $u$.par denote the  index of $u$'s parent  in the tree $b_v$.
		\State Add Vertex($u$.par, $\ind{u}$) to $V_m$.
		\EndFor
		\EndFor
\EndRound

			\end{algorithmic}
	\end{algorithm}

\end{document}